\documentclass[twoside,11pt]{article}
\usepackage{jair, theapa, rawfonts}

\usepackage{times}
\usepackage{graphicx}
\usepackage{latexsym}
\usepackage{comment}
\usepackage[all]{xy}
\usepackage{multirow}
\usepackage{url}

\title{Verification of Agent-Based Artifact Systems}

\author{\name Francesco Belardinelli 
\email
  belardinelli@ibisc.fr\\ \addr Laboratoire Ibisc, Universit\'e
  d'Evry, France\\
\name Alessio Lomuscio \email
  a.lomuscio@imperial.ac.uk\\ \addr Department of Computing, Imperial
  College London, UK\\ \name Fabio Patrizi \email
  fabio.patrizi@dis.uniroma1.it\\ \addr 
  Dipartimento di Ingegneria Informatica, Automatica e Gestionale ``A.~Ruberti''\\
  Universit\`a di Roma ``La Sapienza'', Italy}

\usepackage{xspace}
\usepackage{comment}
\usepackage{latexsym,amssymb,amsfonts}
\usepackage{amsmath}
\usepackage{subfigure}

\usepackage{tikz}
\usetikzlibrary{arrows,snakes,backgrounds,automata,shapes,patterns}

\tikzstyle{every initial by arrow}=[initial text=] 
\tikzstyle{every state}=[ellipse,fill=none,draw=black,text=black,inner sep=1pt,minimum size=2mm] 
\tikzstyle{every picture}=[->,>=stealth',shorten >=1pt,auto,node distance=2cm, semithick]

\newcommand{\aqis}{AC-MAS}
\newcommand{\ACP}{ACP}

\newcommand{\const}{\textit{const}}
\newcommand{\adom}{\textit{adom}}
\newcommand{\im}{Im}
\newcommand{\free}{\textit{free}}
\newcommand{\vars}{\textit{vars}}

\newcommand{\iso}{\simeq}
\newcommand{\Const}{C}
\newcommand{\Vars}{Var}

 \newenvironment{proof}{\textbf{Proof.\ }}{\qed}
\newenvironment{proofsk}{\textbf{Proof (sketch).\ }}{\qed}

\long\def\eatpar#1{%
\ifx#1\par                      
\let\nextmove=\eatpar           
\else
\let\nextmove=#1
\fi
\nextmove
}

\def\qed{\hfill{\qedboxempty}      
  \ifdim\lastskip<\medskipamount \removelastskip\penalty55\medskip\fi}

\def\qedboxempty{\vbox{\hrule\hbox{\vrule\kern3pt
                 \vbox{\kern3pt\kern3pt}\kern3pt\vrule}\hrule}}

\def\qedfull{\hfill{\qedboxfull}   
  \ifdim\lastskip<\medskipamount \removelastskip\penalty55\medskip\fi}

\def\qedboxfull{\vrule height 4pt width 4pt depth 0pt}

\newcommand{{\incolumn}}[1]{\begin{tabular}[c]{c} #1 \end{tabular}}
\newcommand{{\incolumnmath}}[1]{\begin{array}[c]{c} #1 \end{array}}

\begingroup
\catcode`\~=11
\gdef\urltilde{\lower 0.6ex\hbox{~}}
\endgroup

 \newcommand{\D}{\mathcal{D}}
 \newcommand{\F}{\mathcal{F}}

 \renewcommand{\L}{\mathcal{L}}
 
 \renewcommand{\P}{\mathcal{P}}
\newcommand{\Q}{\mathcal{Q}} 
\renewcommand{\S}{\mathcal{S}} \newcommand{\T}{\mathcal{T}}

\newcommand{\defterm}[1]{\mbox{\underline{\it\smash{#1}\vphantom{\lower.1ex\hbox
{x}}}}}

\newcommand{\ra}{\rightarrow}

\newcommand{\lra}{\leftrightarrow}

\newcommand{\set}[1]{\{#1\}}                      

\newcommand{\card}[1]{{|{#1}|}}                     

\newcommand{\tup}[1]{\langle #1\rangle}            

\newcommand{\myi}{{(i)}\xspace}
\newcommand{\myii}{{(ii)}\xspace}
\newcommand{\myiii}{{(iii)}\xspace}
\newcommand{\myiv}{{(iv)}\xspace}

\newcommand{\Nat}{{\rm I\kern-.23em N}}

\newcommand{\commentout}[1]{}

\usepackage{marginnote}
\edef\marginnotetextwidth{\the\textwidth}

\newtheorem{definition}{Definition}[section]
\newtheorem{theorem}[definition]{Theorem}
\newtheorem{lemma}[definition]{Lemma}
\newtheorem{corollary}[definition]{Corollary}

\newtheorem{proposition}[definition]{Proposition}

\begin{document}
\maketitle

\begin{abstract}
Artifact systems are a novel paradigm for specifying and implementing
business processes described in terms of interacting modules called
{\em artifacts}.  Artifacts consist of {\em data} and {\em
lifecycles}, accounting respectively for the relational structure of
the artifacts' states and their possible evolutions over time.
In this paper we put forward artifact-centric multi-agent systems, a
novel formalisation of artifact systems in the context of multi-agent
systems operating on them. Differently from the usual process-based
models of services, the semantics we give explicitly accounts for the
data structures on which artifact systems are defined.

We study the model checking problem for artifact-centric multi-agent
systems against specifications written in a quantified version of
temporal-epistemic logic expressing the knowledge of the agents in the
exchange. 
We begin by noting that the problem is undecidable in general. We then
identify two noteworthy restrictions, one syntactical and one
semantical, that enable us to find bisimilar finite abstractions and
therefore reduce the model checking problem to the instance on finite
models. Under these assumptions we show that the model checking
problem for these systems is EXPSPACE-complete.
We then introduce artifact-centric programs, compact and declarative
representations of the programs governing both the artifact system and
the agents. We show that, while these in principle generate
infinite-state systems, under natural conditions their verification
problem can be solved on finite abstractions that can be effectively
computed from the programs.
Finally we exemplify the theoretical results of the paper through a
mainstream procurement scenario from the artifact systems literature.

\end{abstract}


\section{Introduction}\label{sec:intro}

Much of the work in the area of \emph{reasoning about knowledge}
involves the development of formal techniques for the representation
of epistemic properties of rational actors, or \emph{agents}, in a
multi-agent system (MAS). The approaches based on modal logic are
often rooted on interpreted systems~\cite{ParikhRamanujam85}, a
computationally grounded semantics~\cite{Wooldridge00a} used for the
interpretation of several temporal-epistemic logics. This line of
research was thoroughly explored in the 1990s leading to a significant
body of work~\cite{fhmv:rak}. Further significant explorations have
been conducted since then; a recent topic of interest has focused on
the development of automatic techniques, including model
checking~\cite{cgp99}, for the verification of temporal-epistemic
specifications for the autonomous agents in a
MAS~\cite{mck04,Verics07,LomuscioQuRaimondi09}. This has led to
developments in a number of areas traditionally outside artificial
intelligence, knowledge representation and MAS, including
security~\cite{DechesneWang10,CiobacaDK12},
web-services~\cite{LomuscioSPS10} and cache-coherence
protocols in hardware design~\cite{BaukusMeyden04}.  The ambition of
the present paper is to offer a similar change of perspective in the
area of
\emph{artifact systems}~\cite{CohnH09}, a growing topic in
Service-Oriented Computing (SOC).

\emph{Artifacts} are structures that ``combine data and process in an
holistic manner as the basic building block[s]''~\cite{CohnH09} of
systems' descriptions. Artifact systems are services constituted by
complex workflow schemes based on artifacts which the agents
interact with.  The data component is given by the relational databases
underpinning the artifacts in a system, whereas the workflows are
described by ``lifecycles'' associated with each artifact schema.
While in the standard services paradigm services are made public by
exposing their processes interface, in artifact systems both the data
structures and the lifecycles are advertised. Services are composed in
a ``hub'' where operations on the artifacts are
executed. Implementations of artifact systems, such as the IBM engine
{\sc Barcelona}~\cite{HHV11}, provide a hub where the service
choreography and service orchestratation~\cite{Alonso04} are carried
out.

While artifact systems are beginning to drive new application areas,
such as case management systems~\cite{marin-etal-bpm12}, we identify two
shortcomings in the present state-of-the-art. Firstly, the artifact
systems literature~\cite{BGHLS07,DeutschHPV09,Hull08,NooijenFD2012}
focuses exclusively on the artifacts themselves. While there is
obviously a need to model and implement the artifact infrastructure,
importantly we also  need to account for the agents implementing the
services acting on the artifact system. This is of particular
relevance given that artifact systems are envisaged to play a leading
role in information systems. We need to be able to reason not just
about the artifact states but also about what actions specific
participants are allowed and not allowed to do, what knowledge they
can or cannot derive in a system run, what system state they can
achieve in coordination with their peers, etc. In other words, we need
to move from the description of the artifact infrastructure to one
that encompasses both the agents and the infrastructure.

Secondly, there is a pressing demand to provide the hub with
automatic choreography and orchestration capabilities. It
is well-known that choreography techniques can be leveraged on
automatic model checking techniques; orchestration can be recast as a
synthesis problem, which, in turn, can also benefit from model
checking technology. However, while model checking and its
applications are relatively well-understood in the plain process-based
modelling, the presence of data makes these problems much harder and
virtually unexplored. Additionally, infinite domains in the underlying
databases lead to infinite state-spaces and undecidability of the
model checking problem.

The aim of this paper is to make a concerted contribution to both
problems above. Firstly, we provide a computationally grounded
semantics to systems comprising the artifact infrastructure and the
agents operating on it. We use this semantics to interpret a
temporal-epistemic language with first-order quantifiers to reason
about the evolution of the hub as well as the knowledge of the agents
in the presence of evolving, structured data. We observe that the
model checking problem for these structures is undecidable in general
and analyse two notable decidable fragments. In this context, a
contribution we make is to provide finite abstractions to
infinite-state artifact systems, thereby presenting a technique for
their effective verification for a class of declarative agent-based,
artifact-centric programs that we here define. We evaluate this
methodology by studying its computational complexity and by
demonstrating its use on a well-known scenario from the artifact
systems literature.

\subsection{Artifact-Centric Systems}

Service-oriented computing is concerned with the study and
development of distributed applications that can be automatically
discovered and composed by means of remote interfaces. A point of
distinction over more traditional distributed systems is the
interoperability and connectedness of services and the shared
format for both data and remote procedure calls. Two
technology-independent concepts permeate the service-oriented
literature: orchestration and
choreography~\cite{Alonso04,singh2005service}. Orchestration involves
the ordering of actions of possibly different services, facilitated by
a controller or orchestrator, to achieve a certain overall
goal. Choreography concerns the distributed coordination of different
actions through publicly observable events to achieve a certain
goal. 
A MAS perspective~\cite{Wooldridge09} is known to be particularly
helpful in service-oriented computing in that it allows us to ascribe
information states and private or common goals to the various
services. Under this view the agents of the system implement the
services and interact with one another in a shared infrastructure or
environment.

A key theoretical problem in SOC is to devise effective mechanisms to
verify that service composition is correct according to some
specification. Techniques based on model checking~\cite{cgp99}
and synthesis~\cite{BerardiCGP08} have been put forward to solve the
composition and orchestration problem for services described and
advertised at interface level through finite state
machines~\cite{GLMP08}. More recently, attention has turned
to services described by languages such as
WS-BPEL~\cite{WSBPEL-full}, which provide potentially
unbounded variables in the description of the service process. Again,
model checking approaches have successfully been used to verify complex
service compositions~\cite{BertoliPT10,LomuscioQS12}.

While WS-BPEL provides a model for services with variables, the data
referenced by them is non-permanent. 
The area of data-centric
workflows~\cite{HullNN09,nigam-caswell-03} evolved as an attempt to provide support for
permanent data, 
typically present in the form of underlying databases. 
Although usually abstracted away, permanent 
data is of central importance to services, which
typically query data sources and are driven
by the answers they obtain; see, e.g., \cite{BerardiCGHM05}. 
Therefore, a faithful model of a service behavior 
cannot, in general, disregard this component.
In response to this, proposals have been made in the workflows and
service communities in terms of declarative specifications
of \emph{data-centric services} that are advertised for automatic
discovery and composition. The artifact-centric
approach~\cite{CohnH09} is now one of the leading emerging paradigms
in the area. As described in~\cite{Hull08,Hulletal11}
artifact-centric systems can be presented along four dimensions.

{\bf Artifacts} are the holders of all structured information
available in the system. In a business-oriented scenario this may
include purchase orders, invoices, payment records, etc. Artifacts
may be created, amended, and destroyed at run time; however, abstract
artifact schemas are provided at design time to define the structure
of all artifacts to be manipulated 
in the system. Intuitively,
external events cause changes in the 
system, including 
in the value of artifact attributes.

The evolution of artifacts is governed by {\bf lifecycles}. 
These 
capture the changes that an artifact
may go through from creation to deletion. Intuitively, a purchase order
may be created, amended and operated on by several events before it is
fullfilled and its existence in the system terminated: a lifecycle
associated with a purchase order artifact formalises these
transitions. 

{\bf Services} are seen as the actors operating on the
artifact system. 
They represent both human and 
software actors, possibly distributed, that generate events on the
artifact system. Some services may ``own''
artifacts, and some artifacts may be shared by several services.  
However, not all artifacts, or parts of artifacts, are visible to all services. 
{\em Views} and {\em windows} respectively determine 
which parts of artifacts and which artifact instances
are visible to which service.
 An artifact \emph{hub}
is a 
system that maintains the
artifact system and processes the events generated by the services.

Services generate events on the artifact system according to {\bf
associations}. Typically these are declarative descriptions providing
the precondition and postconditions for the 
generation of events.
These generate changes in the artifact system according to the artifact
lifecycles. Since events may trigger changes in several artifacts in
the system, events are processed by a well-defined
semantics~\cite{damaggio-etal-bpm11,Hulletal11} that governs the sequence of changes an
artifact-system may undertake upon 
consumption of an event. 
Such a semantics, based on the use of 
\emph{Prerequisite-Antecedent-Consequent} (PAC) rules, 
ensures acyclicity and full determinism in the updates on
the artifact system. GSM is a declarative language that can be used to
describe artifact systems. {\sc Barcelona} is an engine that can be
used to run a GSM-based artifact-centric system~\cite{HHV11}.

The above is a partial and incomplete description of the artifact
paradigm. We refer to ~\cite{CohnH09,Hull08,Hulletal11} for more
details.
 
As it will be clear in the next section, in line with the agent-based
approach to services, we will use agent-based concepts to model
services. The artifact-system will be represented as an environment,
constituted by evolving databases, upon which the agents operate;
lifecycles and associations will be modelled by local and global
transition functions. The model is intended to incorporate all
artifact-related concepts including views and windows.

In view of the above in this paper we address the following
questions. How can we give a transition-based semantics for artifacts
and agents operating on them? What syntax should we use to specify
properties of the agents and the artifacts themselves?  Can we verify
that an artifact system satisfies certain properties?  As this will be
shown to be undecidable, can we find suitable fragments on which this
can actually be carried out? If so, what is the resulting complexity?
Lastly, can we provide declarative specifications for the agent
programs so that these can be verified by model checking? Can this
technique be used on mainstream scenarios from the SOC literature?

This paper intends to contribute answering these questions.

\subsection{Related Work}

As stated above, virtually all current literature on artifact-centric
systems focuses on properties and implementations of the
artifact-system as such. Little or no attention is given to the actors
on the system, whether they are human or artificial agents. A few
formal techniques have, however, been put forward to verify the core,
non-agent aspects of the system; in the following we briefly compare
these to this contribution.

To our knowledge the verification of artifact-centric business
processes was first discussed in~\cite{BGHLS07}, where reachability
and deadlocks are phrased in the context of artifact-centric systems
and complexity results for the verification problem are given. The
present contribution differs markedly from~\cite{BGHLS07} by employing
a more expressive specification language, even if the
agent-related aspects are not considered, and by putting forward
effective abstraction procedures for verification.


In~\cite{gerede-su-icsoc07} a verification technique for
artifact-centric systems against a variant of computation-tree logic
is put forward.
The decidability of the verification problem is proven for the language
considered under the assumption that the interpretation domain is
bounded.
Decidability is also shown for the unbounded case by making
restrictions on the values that quantified variables can range over.
In the work here presented we also work on unbounded domains, but do
not require the restrictions present in~\cite{gerede-su-icsoc07}: we
only insist on the fact that the number of distinct values in the
system does not exceed a given threshold at any point in any run. Most
importantly, the interplay between quantification and modalities here
considered allows us to bind and use variables in different
states. This is a major difference as this feature is very expressive
and known to lead to undecidability.

A related line of research is followed
in~\cite{DeutschHPV09,DDV@tods2012}, where the verification problem
for artifact systems against two variants of first-order linear-time
temporal logic is considered.  Decidability of the verification
problem is retained by imposing syntactic restrictions on both the
system descriptions and the specifications to check. This effectively
limits the way in which new values introduced at every computational
step can be used by the system. Properties based on arithmetic
operators are considered in~\cite{DDV@tods2012}.
%
While there are elements of similarity between these approaches and
the one we put forward here, including the fact that the concrete
interpretation domain is replaced by an abstract one, the contribution
here presented has significant differences from these.
Firstly, our setting is branching-time and not linear-time thereby
resulting in different expressive power. Secondly, differently from 
\cite{DeutschHPV09,DDV@tods2012}, we impose no constraints on nested
quantifiers. In contrast, \cite{DDV@tods2012} admits only universal
quantification over combinations of quantifier-free first-order
formulas. 
Thirdly, the abstraction results we present here are given in general
terms on the semantics of declarative programs and do not depend on a
particular presentation of the system. 


More closely related to the present contribution
is \cite{bagheri-etal-corr-2012}, where conditions for the
decidability of the model checking problem for data-centric dynamic
systems, e.g., dynamic systems with relational states, are given. In
this case the specification language used is a first-order version of
the $\mu$-calculus. While our temporal fragment is subsumed by the
$\mu$-calculus, since we use indexed epistemic modalities as well as a
common knowledge operator, the two specification languages have
different expressive power. To retain decidability, like we do here,
the authors assume a constraint on the size of the states. However,
differently from the contribution here
presented, \cite{bagheri-etal-corr-2012} assume limited forms of
quantification whereby only individuals persisting in the system
evolution can be quantified over. In this contribution we do not make
this restriction.

Irrespective of what above, the most important feature that
characterises our work is that the set-up is entirely based on
epistemic logic and multi-agent systems. We use agents to represent
the autonomous services operating in the system and agent-based
concepts play a key role in the modelling, the specifications, and the
verification techniques put forward. Differently from all approaches
presented above we are not only concerned with whether the
artifact-system meets a particular specification. Instead, we also
wish to consider what knowledge the agents in the system acquire by
interacting among themselves and with the artifact-system during a
system run. Additionally, the abstraction methodology put forward is
modular with respect to the agents in the system. These features
enable us to give constructive procedures for the generation of finite
abstractions for artifact-centric programs associated with infinite
models. We are not aware of any work in the literature tackling any
of these aspects.

 

{\bf Relation to previous work by the authors.}  This paper combines
and expands preliminary results originally discussed
in~\cite{BelardinelliLP11},
\cite{BelardinelliLP11b},
 \cite{BelardinelliLP12}, and \cite{BelardinelliLP12b}.  In
particular, the technical set up of artifacts and agents is different
from that of our preliminary studies and makes it more natural to
express artifact-centric concepts such as views. Differently from our
previous attempts we here incorporate an operator for common knowledge
and provide constructive methods to define abstractions for all
notions of bisimulation. We also consider the complexity of the
verification problem, previously unexplored, and evaluate the
technique in detail on a case study.


\subsection{Scheme of the Paper}

The rest of the paper is organised as follows. In
Section~\ref{sec:aqis} we introduce Artifact-centric Multi-Agent
Systems (ACMAS), the semantics we will be using throughout the paper to
describe agents operating on an artifact system. In the same section
we put forward FO-CTLK, a first-order logic with knowledge and time to
reason about the evolution of the knowledge of the agents and the
artifact system. This enables us to propose a satisfaction relation
based on the notion of bounded quantification, define the model
checking problem, and highlight some properties of isomorphic states. 

An immediate result we will explore concerns the undecidability of the
model checking problem for ACMAS in their general
setting. Section~\ref{icsoc} is concerned with synctactical
restrictions on FO-CTLK that enable us to 
guarantee the existence of finite abstractions of infinite-state
ACMAS, thereby making the model checking problem feasible by means of
standard techniques.

Section~\ref{kr} tackles restrictions orthogonal to those of
Section~\ref{icsoc} by focusing on 
a subclass of ACMAS that admits a decidable model checking problem
when considering full FO-CTLK specifications. The key finding here is
that bounded and uniform ACMAS, a class identified by studying a
strong bisimulation relation, admit finite abstractions for any
FO-CTLK specification. The section concludes by showing that under
these restrictions the model checking problem is EXPSPACE-complete.

We turn our attention to artifact programs in Section~\ref{sec:as} by
defining the concept of artifact-centric programs. We define
them through natural, first-order preconditions and postconditions in
line with the artifact-centric approach. We give a semantics to them
in terms of ACMAS and show that their generated models are precisely
those
uniform ACMAS studied earlier in the paper. It follows that,
under some boundedness conditions, which can be naturally expressed, the
model checking problem for artifact-centric programs is
decidable and can be executed on finite models.

Section~\ref{sec:toy} reports a scenario from the artifact systems
literature. This is used to exemplify the technique by providing
finite abstractions that can be effectively verified.

We conclude in Section~\ref{sec:concl} where we consider the
limitations of the approach and point to further work.















\section{Artifact-Centric Multi-Agent Systems}\label{sec:aqis}




In this section we formalise artifact-centric systems and state their
verification problem. As data and databases are important constituents
of artifact systems, our formalisation of artifacts relies on them as
underpinning concepts. However, as discussed in the previous section, we
here give prominence to agent-based concepts. As such, we define our
systems as comprising both the artifacts in the system as well as the
agents that interact with the system. 

A standard paradigm for logic-based reasoning about agent systems
is \emph{interpreted systems}~\cite{ParikhRamanujam85,fhmv:rak}. In
this setting agents are endowed with private local states and evolve
by performing actions according to an individual protocol. As data
play a key part, as well as to allow us to specify properties of the
artifact system, we will define the agents' local states as evolving
database instances. We call this formalisation artifact-centric
multi-agent systems (\aqis). \aqis\ enable us to represent naturally
and concisely concepts much used in the artifact paradigm such as the
one of \emph{view} discussed earlier.

Our specification language will include temporal-epistemic logic but
also quantification over a domain so as to represent the data. This is an
usual verification setting, so we will formally define the model
checking problem for this set up.



\subsection{Databases and First-Order Logic}

As discussed above, we use databases as the basic building blocks for defining
the states of the agents and the artifact system. We here fix the
notation and terminology used. We refer to~ \cite{AbiteboulHV95} for
more details on databases.
\begin{definition}[Database Schemas] \label{def:dbschema}
  A {\em (relational) database schema} is a set $\D=\set{P_1/q_1,
    \ldots,P_n/q_n}$ of {\em relation 
    symbols} $P_i$, each associated with its arity $q_i\in\mathbb{N}$.
\end{definition}

Instances of database schemas are defined over 
interpretation domains.

\begin{definition}[Database Instances]
Given an interpretation domain $U$ and a database schema $\D$, a {\em
$\D$-instance} over $U$  is a mapping $D$ associating each relation
symbol $P_i \in \D$ with a {\em finite} $q_i$-ary relation over $U$,
i.e., $D(P_i)\subseteq U^{q_i}$.
\end{definition}

The set of all $\D$-instances over an interpretation domain $U$ is
denoted by $\D(U)$. We simply refer to ``instances'' whenever the
database schema $\D$ is clear by the context.  The {\em active domain}
of an instance $D$, denoted as $\adom(D)$, is the set of all
individuals in $U$ occurring in some tuple of some predicate
interpretation $D(P_i)$.  Observe that, since $\D$ contains a finite
number of relation symbols and each $D(P_i)$ is finite, so is
$\adom(D)$.

To fix the notation, we recall the syntax of first-order formulas with
equality and no function symbols.  Let $\Vars$ be a countable set of
{\em individual variables} and $\Const$ be a finite set of {\em
individual constants}.  A \emph{term} is any element
$t\in \Vars\cup \Const$.
\begin{definition}[FO-formulas over $\D$]\label{def:fo}
Given a database schema $\D$, the formulas $\varphi$ of the
first-order language $\L_{\D}$ are defined by the following BNF
grammar:
\begin{eqnarray*}
\varphi & ::= & t=t'\mid P_i(t_1, \ldots ,t_{q_i})\mid \lnot \varphi \mid \varphi \ra \varphi \mid \forall x \varphi 
\end{eqnarray*}
where $P_i \in \D$, $t_1, \ldots ,t_{q_i}$ is a $q_i$-tuple of terms and $t,t'$ are terms.
\end{definition}
We assume ``$=$'' to be a special binary predicate with fixed obvious
interpretation. To summarise, $\L_\D$ is a first-order language with
equality over the relational vocabulary $\D$ with no function symbols
and with finitely many constant symbols from $\Const$.
Observe that considering a finite set of constants is not a
limitation. Indeed, since we will be working with finite sets of
formulas, $\Const$ can always be defined so as to be able to
express any formula of interest.

In the following we use the standard abbreviations $\exists$, 
$\wedge$,
$\vee$, and $\neq$.
Also, free and bound variables are defined as standard.
For a formula $\varphi$ we denote the set of its variables as
$\vars(\varphi)$, the set of its free variables as $\free (\varphi)$,
and the set of its constants
as $\const(\varphi)$.  We write
$\varphi(\vec x)$ 
to list
explicitly in arbitrary order all the free variables
$x_1,\ldots,x_\ell$ of $\varphi$.  By slight abuse of notation, we
treat $\vec x$ as a set, thus we write $\vec x=\free(\varphi)$.
A {\em sentence} is a formula with no free variables.  

Given an interpretation domain $U$ such that $\Const\subseteq U$,
an {\em assignment} is a function $\sigma: \Vars \mapsto U$.
For an assignment $\sigma$, we denote by $\sigma{x\choose u}$ the 
assignment such that:
\myi $\sigma{x\choose u}(x) = u$; and \myii 
$\sigma{x\choose u}(x') = \sigma(x')$, for every $x' \in \Vars$
different from $x$.
For convenience, we extend assignments to constants so that
$\sigma(t)=t$, if $t\in \Const$; that is, we assume a Herbrand
interpretation of constants.
We can now define the semantics of $\L_\D$.
\begin{definition}[Satisfaction of FO-formulas]\label{def:fo-sem}
Given a $\D$-instance $D$, an assignment $\sigma$,
and an FO-formula $\varphi\in\L_{\D}$, we inductively define
whether $D$ \emph{satisfies $\varphi$ under $\sigma$}, written $ (D, \sigma)
\models \varphi$, as follows:
\begin{tabbing}
 $ (D, \sigma)\models P_i(t_1,\ldots,t_{q_i})$ \ \ \ \ \ \ \ \=
 iff \ \ \ \ \ \ \= $\tup{\sigma(t_1),\ldots,\sigma(t_{q_i})}\in D(P_i)$\\ 
 $ (D, \sigma)\models t=t'$ \> iff \> $\sigma(t)=\sigma(t')$\\
 $ (D, \sigma)\models \lnot\varphi$ \> iff \> it is not the case that $ (D, \sigma) \models\varphi$\\
 $ (D, \sigma)\models \varphi\ra\psi$ \> iff \> $(D,\sigma)\models\lnot\varphi $ or $ (D, \sigma)\models\psi$\\
 $ (D, \sigma)\models \forall x\varphi$ \> iff \> for all
 $u\in \adom(D)$, we have that $ (D, \sigma{x \choose u}) \models\varphi$
\end{tabbing}

A formula $\varphi$ is {\em true} in $D$, written $D\models\varphi$,
iff $ (D, \sigma) \models \varphi$, for all assignments $\sigma$.
\end{definition}
Observe that we adopt an {\em active-domain} semantics, that is,
quantified variables range only over the active domain of $D$.
Also notice that constants are interpreted rigidly; so,
two constants are equal if and only if they are syntactically the
same.  In the rest of the paper, we assume that every interpretation
domain includes $\Const$. Also, as a usual shortcut, we write
$(D, \sigma)\not \models \varphi$ to express that it is not the case
that $(D, \sigma) \models \varphi$.

Finally, we introduce the $\oplus$ operator on $\D$-instances
that will be used later in the paper.
Let the {\em primed version} of a database schema $\D$ 
be the schema $\D'=\set{P'_1/q_1, \ldots,P'_n/q_n}$ obtained from 
$\D$ by syntactically replacing each predicate symbol $P_i$ 
with its {\em primed version} $P_i'$ of the same arity.
\begin{definition}[$\oplus$ Operator]
Given two $\D$-instances $D$ and $D'$, 
we define $D\oplus D'$ as the $(\D\cup\D')$-instance such that
$D\oplus D'(P_i)=D(P_i)$ and $D\oplus D'(P'_i)=D'(P_i)$.
\end{definition}
Intuitively, the $\oplus$ operator defines a disjunctive join of the
two instances, where relation symbols in $\D$ are interpreted according
to $D$, while their primed versions are interpreted according to $D'$.


\subsection{Artifact-Centric Multi-Agent Systems}\label{subsec:acmas}

In the following we introduce the semantic structures that we will use
throughout the paper. We define an artifact-centric multi-agent
system as a system comprising an environment representing all
interacting artifacts in the system and a finite set of agents
interacting with such environment. As agents have views of the
artifact state, i.e., projections of the status of particular
artifacts, we assume the building blocks of their private local states
also to be modelled as database instances. In line with the interpreted
systems semantics~\cite{fhmv:rak} not everything in the agents' states
needs to be present in the environment; a portion of it may be
entirely private and not replicated in other agents' states.
%
So, we start by introducing the notion of \emph{agent}.
\begin{definition}[Agent]\label{def:ags}
Given an interpretation domain $U$, an {\em agent} is a tuple
$A=\tup{\D,L,Act, Pr}$, where:
\begin{itemize}
\item $\D$ is the {\em local database schema};
\item $L\subseteq \D(U)$ is the set of {\em local states};
\item $Act$ is the finite set of \emph{action types} of the form  
$\alpha(\vec p)$, where $\vec p$ is the tuple of \emph{abstract
parameters};
\item $Pr:L\mapsto 2^{Act(U)}$ is the {\em local protocol
function}, where $Act(U)$ is the set of \emph{ground actions} of the
form $\alpha(\vec u)$ where $\alpha(\vec p)\in Act$ and $\vec u \in
U^{|\vec p|}$ is a tuple of \emph{ground parameters}.
\end{itemize}
\end{definition}

Intuitively, at a given time each agent $A$ is in some local state
$l \in \D(U)$ that represents all the information agent $A$ has at its
disposal. In this sense we follow~\cite{fhmv:rak} but require
that this information is structured as a database. Again, following
standard literature we assume that the agents are autonomous and
proactive and perform the actions in $Act$ according to the protocol
function $Pr$. In the definition above we distinguish between
``abstract parameters'' to denote the language in which particular
action parameters are given, and their concrete values or
``ground parameters''.

We assume that the agents interact among themselves and with an
environment comprising all artifacts in the system. As artifacts are
entities involving both data and process, we can see them as
collections of database instances paired with actions and governed by
special protocols. Without loss of generality we can assume the
environment state to be a single database instance including all
artifacts in the system. From a purely formal point of view this
allows us to represent the environment as a special agent. Of course,
in any specific instantiation the environment and the agents will be
rather different, exactly in line with the standard propositional
version of interpreted systems.

We can therefore define the synchronous composition of agents with the
 environment.
\begin{definition}[Artifact-Centric Multi-Agent Systems]\label{def:sys-db-ag}
Given an interpretation domain $U$ and a set $Ag=\set{A_0,\ldots,A_n}$
of agents $A_i=\tup{\D_i,L_i,Act_i, Pr_i}$ defined on $U$, an {\em
artifact-centric multi-agent system} (or \aqis) is a tuple
$\P=\tup{\S, U, s_0,\tau}$ where:
\begin{itemize}
\item $\S\subseteq L_0\times\cdots\times L_n$ is the set of {\em
reachable global states};
\item $U$ is the {\em interpretation domain};
\item $s_0\in\S$ is the {\em initial global state};
\item  $\tau:\S\times Act(U) \mapsto 2^\S$
	is the {\em global transition function}, where
	$Act(U)=Act_0(U)\times\cdots\times Act_n(U)$ is the set
	of \emph{global (ground) actions}, and
	$\tau(\tup{l_0,\ldots,l_n},\tup{\alpha_0(\vec
	u_0),\ldots,\alpha_n(\vec u_n)})$ is defined iff
	$\alpha_i(\vec u_i) \in Pr_i(l_i)$ for every $i \leq
	n$.  \end{itemize}
\end{definition}

As we will see in later sections, \aqis\ are the natural extension of
interpreted systems to the first order to account for environments
constituted of artifact-centric systems. They can be seen as a
specialisation of quantified interpreted
systems~\cite{BelardinelliL12}, a general extension of
interpreted systems to the first-order case.

In the formalisation above the agent $A_0$ is referred to as the {\em
environment} $E$. The environment includes all artifacts in the system
as well as additional information to facilitate communication between the
agents and the hub, e.g., messages in transit etc.  At any given time
an \aqis\ is described by a tuple of database instances, representing
all the agents in the system as well as the artifact system. A single
interpretation domain for all database schemas is given. Note that
this does not break the generality of the representation as we can
always extend the domain of all agents and the environment before
composing them into a single \aqis. The global transition function
defines the evolution of the system through synchronous composition of
actions for the environment and all agents in the system.

Much of the interaction we are interested in modelling involves
message exchanges with payload, hence the action parameters, between
agents and the environment, i.e., agents operating on the
artifacts. However, note that the formalisation above does not
preclude us from modelling agent-to-agent interactions, as the global
transition function does not rule out successors in which only some
agents change their local state following some actions. Also observe
that essential concepts such as \emph{views} are naturally expressed
in \aqis\ by insisting that the local state of an agent includes part
of the environment's, i.e., the artifacts the agent has access to. Not
all \aqis\ need to have views defined, so it is also possible for the
views to be empty.

Other artifact-based concepts such as lifecycles are naturally
expressed in \aqis. As artifacts are modelled as part of the
environment, a lifecycle is naturally encoded in \aqis\ simply as the
sequence of changes induced by the transition function $\tau$ on the
fragment of the environment representing the lifecycle in question. We
will show an example of this in Section~\ref{sec:toy}. 
  
Some technical remarks now follow. To simplify the notation, we denote
a global ground action as $\vec\alpha(\vec u)$, where
$\vec \alpha=\tup{\alpha_0(p_0),\ldots,\alpha_n(p_n)}$ and $\vec
u=\tup{\vec u_0,\ldots,\vec u_n}$, with each $\vec u_i$ of appropriate
size.
We define the {\em transition relation} $\to$ on $\S \times \S$ such
that $s \to s'$ if and only if there exists a $\vec \alpha(\vec u) \in
Act(U)$ such that $s'\in\tau(s,\vec \alpha(\vec u))$.
If $s \to s'$, we say that $s'$
is a \emph{successor} of $s$.
A \emph{run} $r$ from $s\in\S$ is an infinite
sequence $s^0 \to s^1\to \cdots $,  with $s^0=s$. 
For $n\in \mathbb{N}$, we take $r(n)\doteq s^n$.  
A state $s'$ is 
\emph{reachable from $s$} if there exists a run $r$ from 
the global state 
$r(0) = s$ such that~$r(i) = s'$, for some $i\geq 0$.
We assume that the relation $\ra$ is serial. This can be easily
obtained by assuming that each agent has a \textsf{skip} action
enabled at each local state and that performing \textsf{skip} induces
no changes in any of the local states. We consider $\S$ to be the set
of states reachable from the initial state $s_0$.
For convenience we will use also the concept of {\em
temporal-epistemic} (t.e., for short) run.  Formally a t.e.~run $r$
from a state $s \in \S$ is an infinite sequence $s^0\leadsto
s^1\leadsto \ldots$ such that $s^0=s$ and $s^i \to s^{i+1}$ or
$s^i \sim_k s^{i+1}$, for some $k \in Ag$.  
A state $s'$ is said to be \emph{temporally-epistemically reachable}
(t.e.~reachable, for short) from $s$ if there exists a t.e.~run $r$
from the global state $r(0) = s$ such that for some $i\geq 0$
we have that $r(i) = s'$.  Obviously, temporal-epistemic runs include
purely temporal runs as a special case.

As in plain interpreted systems~\cite{fhmv:rak}, we say that two
global states $s=\tup{l_0,\ldots,l_n}$ and $s'=\tup{l'_0,\ldots,l'_n}$
are
\emph{epistemically indistinguishable} 
for agent $A_i$, written $s \sim_i s'$, 
if $l_i = l'_i$. Differently from interpreted systems the local equality is
evaluated on database instances. 
%
Also, notice that we admit $U$
to be infinite, thereby allowing the possibility of 
the set of states $\S$ to be infinite.  Indeed, unless we specify
otherwise, we will assume to be working with infinite-state \aqis.

Finally, for technical reasons it is useful to refer to a {\em global}
database schema $\D = \D_0\cup\cdots\cup \D_n$ of an \aqis.
Every global state $s=\tup{l_0,\ldots,l_n}$ is associated with the
(global) $\D$-instance $D_s\in\D(U)$ such that
$D_s(P_i)=\bigcup_{j \in Ag} l_j(P_i)$, for $P_i\in \D$.
We omit the subscript $s$ when $s$ is clear from the context and we
write $\adom(s)$ for $\adom(D_s)$.
Notice that for every $s \in \S$, the $D_s$ associated
with $s$ is unique, while the converse is not true in general.
\label{refref}

\subsection{Model Checking}
We now define the problem of verifying an artifact-centric multi-agent
system against a specification of interest.
%
By following the artifact-centric model, we wish to give data the same
prominence as processes. To deal with data and the underlying database
instances, our specification language needs to include first-order
logic. Further, we require temporal logic to describe the system
execution. Lastly, we use epistemic logic to express the information the
agents have at their disposal. Hence, we define a first-order temporal
epistemic specification language to be interpreted on AC-MAS. The
specification language will be used in Section~\ref{sec:as} to formalise
properties of artifact-centric programs.

\begin{definition}[The Logic FO-CTLK] \label{FOCTLK}
The first-order CTLK (or FO-CTLK) formulas  $\varphi$ over a database
schema $\D$ are inductively defined by the following BNF:
\begin{eqnarray*}
\varphi & ::= & \phi \mid \lnot \varphi \mid
\varphi \ra \varphi \mid \forall x \varphi \mid AX \varphi \mid A
\varphi U \varphi \mid E \varphi U \varphi \mid K_i \varphi \mid C \varphi
\end{eqnarray*}
where $\phi\in\L_{\D}$ and $0 < i \leq n$.
\end{definition}
The notions of free and bound variables for FO-CTLK 
 extend straightforwardly from $\L_\D$, as well as functions 
$\vars$, $\free$, and $\const$. 
As usual, the temporal formulas $AX \varphi$ and $A
\varphi U \varphi'$ (resp.~$E \varphi U \varphi'$) are read as ``for
all runs, at the next step $\varphi$'' and ``for all runs (resp.~some
run), $\varphi$ until $\varphi'$''.  The epistemic formulas
$K_i \varphi$ and $C \varphi$ intuitively mean that ``agent $A_i$
knows $\varphi$'' and ``it is common knowledge among all agents that
$\varphi$'' respectively. 
We use the abbreviations $EX\varphi$, $AF\varphi$, $AG
\varphi$, $EF\varphi$, and $EG\varphi$ as standard.
Observe that free variables can occur within the scope of modal
operators, thus allowing for the unconstrained alternation of
quantifiers and modal operators, thereby allowing us to refer to
elements in different modal contexts.
We consider also a number of fragments of FO-CTLK. 
%
The {\em sentence atomic} version of FO-CTLK without epistemic
modalities, or SA-FO-CTL, is the language obtained from
Definition~\ref{FOCTLK} by removing the clauses for epistemic operators and
restricting atomic formulas to first-order \emph{sentences}, so that
no variable appears free in the scope of a modal operator:
\begin{eqnarray*}
\varphi & ::= & \phi \mid
\lnot\varphi\mid
\varphi \ra \varphi \mid
AX \varphi \mid A
\varphi U \varphi \mid E \varphi U \varphi 
\end{eqnarray*}
where $\phi\in\L_{\D}$ is a sentence.

We will consider also the language FO-ECTLK, i.e., the
existential fragments of FO-CTLK, defined as follows:
\begin{eqnarray*}
\varphi & ::= & \phi \mid \varphi\wedge\varphi \mid \varphi\vee\varphi\mid
\forall x \varphi \mid \exists x \varphi \mid EX \varphi \mid 
E \varphi U \varphi\mid \bar K_i \varphi \mid \bar C \varphi,
\end{eqnarray*}
where $\phi\in\L_\D$, 
with $\wedge$ and $\vee$ the standard abbreviations,
$\bar K_i \varphi\equiv \lnot K_i \lnot\varphi$, and 
$\bar C \varphi\equiv \lnot C \lnot\varphi$.



The semantics of FO-CTLK formulas is defined as follows.
\begin{definition}[Satisfaction for FO-CTLK]\label{def:fo-ctl-sem}
Consider an \aqis\ $\P$,
an FO-CTLK formula $\varphi$, 
a state $s \in \P$, and an assignment $\sigma$.
We inductively define whether $\P$ {\em satisfies} $\varphi$ 
in $s$ under $\sigma$, 
written $(\P,s,\sigma)\models \varphi$, as follows:
\begin{tabbing}
	 $(\P,s,\sigma)\models \varphi\ra\varphi'$ \ \ \ \ \ \ \=
	 iff \ \ \ \ \ \ \ \=  $(\P,D_s,\sigma)\not\models\varphi$ or $(\P,D_s,\sigma)\models\varphi'$\kill
	 $(\P,s,\sigma)\models \varphi$  \> iff \>
         $(D_s,\sigma) \models \varphi$, if $\varphi$ is an FO-formula\\
	 $(\P,s,\sigma)\models \lnot\varphi$ \>  iff \>  it is not the
	 case that $(\P,s,\sigma)\models\varphi$\\
	 $(\P,s,\sigma)\models \varphi\ra\varphi'$  \> iff \>  $(\P,s,\sigma)\models\lnot\varphi$ or $(\P,s,\sigma)\models\varphi'$\\
	 $(\P,s,\sigma)\models \forall x \varphi$  \> iff \>  for all $u \in \adom(s)$, $(\P,s,\sigma{x \choose u}) \models \varphi$\\
	 $(\P,s,\sigma)\models AX\varphi$ \>  iff \>  for all runs $r$,   if $r(0)=s$, then $(\P,r(1),\sigma)\models\varphi$\\
	 $(\P,s,\sigma)\models A\varphi U\varphi'$ \>  iff \>  for all
	 runs $r$, if $r(0)=s$, then there is $k\geq 0$ s.t.~$(\P,r(k),\sigma)
\models\varphi'$,\\ \> \> and for all $j$, $0\leq j< k$ implies $(\P,r(j),\sigma)\models\varphi$\\
	 $(\P,s,\sigma)\models E\varphi U\varphi'$ \>  iff \>  for
	 some run $r$, $r(0)=s$ and there is $k\geq 0$ s.t.~$(\P,r(k),\sigma)
\models\varphi'$,\\ \> \>  and for all $j$, $0\leq j< k$ implies $(\P,r(j),\sigma)\models\varphi$\\
	 $(\P,s,\sigma)\models K_i \varphi$ \>  iff \>  
for all $s'$, $s \sim_i s'$ implies $(\P,s',\sigma)\models\varphi$\\
	 $(\P,s,\sigma)\models C \varphi$ \>  iff \>  
for all $s'$, $s \sim s'$ implies $(\P,s',\sigma)\models\varphi$
\end{tabbing}
where $\sim$ is the transitive closure of $\bigcup_{1\dots n} \sim_i$.
\end{definition}
A formula $\varphi$ is said to be \emph{true} at a state $s$, written $(\P, s) \models
\varphi$, if $(\P,s,\sigma)\models \varphi$ for all assignments $\sigma$. Moreover,
$\varphi$ is said to be {\em true} in $\P$, written
$\P \models \varphi$, if $(\P, s_0) \models \varphi$.

A key concern in this paper is to explore the model checking of \aqis\
against first-order temporal-epistemic specifications.
\begin{definition}[Model Checking]
Model checking an \aqis\ $\P$ against an FO-CTLK formula $\varphi$
amounts to finding an assignment $\sigma$ such that
$(\P,s_0, \sigma) \models\varphi$.
\end{definition}
It is easy to see that whenever $U$ is finite the model checking problem
is decidable as $\P$ is a finite-state system. In general this is not
the case. 
\begin{theorem}\label{th:undec}
The model checking problem for \aqis\ w.r.t.~FO-CTLK is undecidable.
\end{theorem}
\begin{proofsk}
This can be proved by showing that every Turing machine $T$ whose tape
contains an initial input $I$ can be simulated by an artifact system
$\P_{T,I}$. The problem of checking whether $T$ terminates on
that particular input can be reduced to checking whether
$\P_{T,I}\models\varphi$, where $\varphi$ encodes the termination
condition.  The detailed construction is similar to that of
Theorem~4.10 of \cite{DeutschSV07}.
\end{proofsk}
Given the general setting in which the model checking problem is
defined above, the negative result is not surprising. In the following
we identify syntactic and semantic restrictions for which the problem
is decidable.
%

\subsection{Isomorphisms}

We now investigate the concept of isomorphism on \aqis. This will be
needed in later sections to produce finite abstractions of
infinite-state \aqis.  In what follows let $\P = \tup{\S, U,
s_{0}, \tau}$ and $\P' = \tup{\S', U', s'_{0}, \tau}$ be two \aqis.
%

\begin{definition}[Isomorphism] \label{def:ad-iso}
Two local states $l, l'\in \D(U)$ are {\em isomorphic}, written
$l \iso l'$, iff there exists a bijection
$\iota:\adom(l)\cup \Const \mapsto \adom(l')\cup\Const$ such that:
\begin{itemize}
\item[\myi] 
$\iota$ is the identity on $\Const$;  
\item[\myii] 
for every $P_i \in \D$, $\vec u \in U^{q_i}$, 
we have that
$\vec{u} \in l(P_i)$ iff $\iota(\vec u) \in l'(P_i)$.
\end{itemize}
When this is the case, we say that $\iota$ is a \emph{witness} for
$l \iso l'$.

Two global states $s \in \S$ and $s' \in \S'$ are {\em isomorphic}, written
$s \iso s'$, iff there exists a bijection
$\iota:\adom(s)\cup \Const \mapsto \adom(s')\cup\Const$ such that
for every 
$j \in Ag$, $\iota$ is a witness for $l_j \iso l'_j$.
\end{definition}
Notice that isomorphisms preserve the constants in $\Const$ 
as well as predicates in the local states up to renaming of the
corresponding terms. 
%
%
Any function $\iota$ as above is called a \emph{witness} for $s \iso
s'$.  Obviously, the relation $\iso$ is an equivalence relation.
Given a function $f: U \mapsto U'$ defined on $\adom(s)$,
$f(s)$ denotes the interpretation in $\D(U')$
obtained from $s$ by renaming each
$u \in \adom(s)$ as $f(u)$.
If $f$ is also injective (thus invertible) and the identity on
$\Const$, then $f(s) \iso s$.


{\bf Example.}
For an example of isomorphic states, consider an 
agent with local database schema $\D=\set{P_1/2, P_2/1}$,
let $U=\set{a,b,c,\ldots}$ be an interpretation domain, and
fix the set $C=\set{b}$ of constants. Let $l$ be the local state
such that 
$l(P_1)=\set{\tup{a,b},\tup{b,d}}$ and 
$l(P_2)=\set{a}$ (see Figure~\ref{fig:iso}).  
Then, the local state $l'$
such that 
$l'(P_1)=\set{\tup{c,b},\tup{b,e}}$ and 
$l'(P_2)=\set{c}$ is isomorphic to $l$. 
This can be easily seen by considering 
the isomorphism $\iota$, where:
$\iota(a)=c$, $\iota(b)=b$, and $\iota(d)=e$.
On the other hand, the state $l''$ where
$l''(P_1)=\set{\tup{f,d},\tup{d,e}}$ and
$l''(P_2)=\set{f}$ is not isomorphic to $l$. 
Indeed,  although a bijection exists that 
``transforms'' $l$ into $l''$, it is easy to see that  
none can be such that  $\iota'(b)=b$.
\begin{figure}
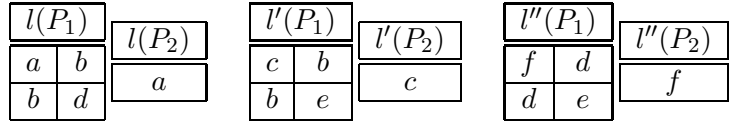

	\centering
	\begin{tabular}{|c|c|}
	\hline
	\multicolumn{2}{|c|}{$l(P_1)$}\\
	\hline\hline
	$a$ & $b$ \\
	\hline
	$b$ & $d$ \\
	\hline
	\end{tabular}
	\begin{tabular}{|c|}
	\hline
	$l(P_2)$\\
	\hline\hline
	$a$ \\
	\hline
	\end{tabular}
	~~~~
	\begin{tabular}{|c|c|}
	\hline
	\multicolumn{2}{|c|}{$l'(P_1)$}\\
	\hline\hline
	$c$ & $b$ \\
	\hline
	$b$ & $e$ \\
	\hline
	\end{tabular}
	\begin{tabular}{|c|}
	\hline
	$l'(P_2)$\\
	\hline\hline
	$c$ \\
	\hline
	\end{tabular}
	~~~~
	\begin{tabular}{|c|c|}
	\hline
	\multicolumn{2}{|c|}{$l''(P_1)$}\\
	\hline\hline
	$f$ & $d$ \\
	\hline
	$d$ & $e$ \\
	\hline
	\end{tabular}
	\begin{tabular}{|c|}
	\hline
	$l''(P_2)$\\
	\hline\hline
	$f$ \\
	\hline
	\end{tabular}
	\caption{Examples of isomorphic and non-isomorphic local states.\label{fig:iso}}
\end{figure}


Note that, while isomorphic states have the same relational structure,
two isomorphic states do not necessarily satisfy the same FO-formulas
as satisfaction depends also on the values assigned to free variables.
To account for this, we introduce the following notion.
\begin{definition}[Equivalent assignments]\label{def:eq-as}
Given two
states $s\in \S$ and $s'\in \S'$, and a set of variables $V \subseteq
Var$,
two assignments $\sigma : Var \mapsto U$ and $\sigma' : Var \mapsto U'$ are 
\emph{equivalent for $V$ w.r.t.~$s$ and $s'$}
 iff 
there exists a bijection
$\gamma:\adom(s)\cup \Const\cup \sigma(V)
\mapsto\adom(s')\cup \Const\cup \sigma'(V)$
such that:
\begin{itemize}
\item[\myi] 
$\gamma|_{\adom(s)\cup \Const}$ is a witness for $s \iso s'$; 
\item[\myii] 
$\sigma'|_V= \gamma \circ \sigma|_V$.
\end{itemize}
\end{definition}
Intuitively, equivalent assignments preserve both the (in)equalities of
the variables in $V$ and
the constants in $s, s'$ up to renaming. Note that, by definition, the
above implies that $s,s'$ are isomorphic.
We say that two assignments are \emph{equivalent for an FO-CTLK
formula} $\varphi$, omitting the states $s$ and $s'$ when it is clear
from the context, if these are equivalent for $\free(\varphi)$.

We can now show that isomorphic states satisfy exactly the same
FO-formulas.
\begin{proposition}\label{prop:iso-inst-eq}
Given two isomorphic states $s\in \S$ and $s'\in \S'$, an FO-formula
$\varphi$, and two assignments $\sigma$ and $\sigma'$ equivalent for
$\varphi$, we have that
\begin{eqnarray*}
	(D_{s}, \sigma) \models \varphi & \text{iff} &
	(D_{s'}, \sigma') \models \varphi
\end{eqnarray*}
\end{proposition}

\begin{proof}
The proof is by induction on the structure of $\varphi$.  Consider the
base case for the atomic formula $\varphi \equiv P(t_1, \ldots, t_k)$.
Then $(D_s,\sigma) \models \varphi$ iff
$\langle \sigma(t_1), \ldots, \sigma(t_k) \rangle \in D_s(P)$. Since
$\sigma$ and $\sigma'$ are equivalent for $\varphi$, and $s \iso s'$,
this is the case iff
$\langle \sigma'(t_1), \ldots, \sigma'(t_k) \rangle \in D_{s'}(P)$,
that is, $(D_{s'},\sigma') \models \varphi$. The base case for
$\varphi\equiv t=t'$ is proved similarly, by observing that the
satisfaction of $\varphi$ depends only on the assignments, and that
the function $\gamma$ of Def.~\ref{def:eq-as} is a bijection, thus all
the (in)equalities between the values assigned by $\sigma$ and
$\sigma'$ are preserved. This is sufficient to guarantee that
$\sigma(t)=\sigma(t')$ iff $\sigma'(t)=\sigma'(t')$.
        The inductive step for the propositional connectives is
        straightforward.
	Finally, if $\varphi\equiv \forall y \psi$, then
        $(D_s,\sigma) \models \varphi$ iff for all
        $u \in \adom(s)$, $(D_s,\sigma{y \choose
        u}) \models \psi$. Now
	consider the witness  $\iota = \gamma|_{\adom(s) \cup \Const}$ for
        $s \iso s'$, where $\gamma$ is as in Def.~\ref{def:eq-as}. 
        We have that 
	$\sigma{y \choose u}$ and $\sigma'{y \choose \iota(u)}$
	are equivalent for $\psi$.
By induction hypothesis
	$(D_s,\sigma{y \choose u}) \models \psi$ iff
	$(D_{s'},\sigma'{y \choose \iota(u)}) \models \psi$. Since
	$\iota$ is a bijection, this is the case iff 
	for all $u'\in\adom(s')$, $(D_{s'},\sigma'{y \choose u'}) \models \psi$, 
	i.e.,
	$(D_{s'},\sigma') \models \varphi$.  
\end{proof}

This leads us to the following result.
\begin{corollary} \label{prop:iso-inst-eq2}
Given two isomorphic states $s\in \S$ and $s'\in \S'$ and an FO-sentence
	$\varphi$,
we have that
\begin{eqnarray*}
	D_{s} \models \varphi & \text{iff} &
	D_{s'} \models \varphi
\end{eqnarray*}
\end{corollary}

\begin{proof} From right to left. Suppose, by contradiction, that
$D_{s} \not \models \varphi$. Then there exists an assignment $\sigma$
  s.t.~$(D_{s},\sigma) \not \models \varphi$.  Since $\free(\varphi)
  = \emptyset$, if $\iota$ is a witness for $s \simeq s'$, then the
  assignment $\sigma' = \iota \circ \sigma$ is equivalent to $\sigma$
  for $s$ and $s'$.  By Proposition~\ref{prop:iso-inst-eq} we have
  that $(D_{s'},\sigma') \not \models \varphi$, that is,
$D_{s'} \not \models \varphi$. The case from left to right can be shown
similarly.
\end{proof}

Thus, isomorphic states cannot be distinguished by FO-sentences. This
enables us to use this notion when defining simulations as we will see
in the next section.

\section{Abstractions for Sentence Atomic FO-CTL} \label{icsoc}

In the previous section we have observed that model checking
\aqis\ against FO-CTLK is undecidable in general. So, it is clearly of interest to 
identify decidable settings. In what follows we introduce two main
results. The first, presented in this section, identifies restrictions
on the language; the second, presented in the next section, focuses on
semantic constraints. While these cases are in some sense orthogonal to
each other, we show that they both lead to decidable model checking
problems. They are also both carried out on a rather natural subclass
of \aqis\ that we call \emph{bounded}, which we identify below. Our
goal for proceeding in this manner is to identify finite abstractions
of infinite-state \aqis\ so that verification of programs, that
admit \aqis\ as models, can be conducted on them, rather than on
infinite-state \aqis . We will see this in detail in
Section~\ref{sec:as}.


Given our aims we begin by defining a first notion of bisimulation in
the context of \aqis. Bisimulations will be used to show that all
bounded \aqis\ admit a finite, bisimilar, abstraction that satisifies
the same SA-FO-CTL specifications as the original \aqis .  Also in
what follows we assume that $\P=\tup{\S, U, s_0,\tau}$ and
$\P'=\tup{\S', U', s'_0,\tau'}$.




\begin{definition}[Simulation]\label{def:SA-sim}
A relation $R \subseteq \S \times \S'$ is a \emph{simulation}
iff
$\tup{s,s'}\in R$ implies:
\begin{enumerate} 
	\item\label{SA-simul-1} $s \iso s'$;
	\item\label{SA-simul-2} for every $t \in \S$, if $s \ra t$
	then there exists $t' \in \S'$ 
	s.t.~$s'\ra t'$   and $\tup{t, t'}\in R$.
      	\end{enumerate}
\end{definition}
Definition~\ref{def:SA-sim} presents the standard notion of simulation
applied to the case of \aqis. The difference from the propositional
case is that we here insist on the states being isomorphic, a
generalisation from the usual requirement for propositional valuations
to be equal~\cite{BlackburndRV01}. 
As in the standard case, two states $s\in \S$ and 
$s'\in \S'$ are said to be \emph{similar}, written $s\preceq s'$, 
if there exists a simulation relation $R$ s.t.~$\tup{s,s'}\in R$.
It can be proven that the similarity relation $\preceq$ 
is a simulation itself, and in particular the largest one w.r.t. set inclusion,
and that it is transitive and reflexive.
Finally, we say that $\P'$ {\em simulates} $\P$, written $\P\preceq \P'$, 
if $s_0\preceq s'_0$. We extend the above to bisimulations.

\begin{definition}[Bisimulation]\label{def:SA-bisim-ext}
A relation $B \subseteq \S \times \S'$ is a \emph{bisimulation}
iff both $B $ and $B^{-1}=\set{\tup{s',s}\mid \tup{s,s'}\in B}$ are 
simulations.
\end{definition}
We say that two states $s\in \S$ and $s'\in \S'$ are 
{\em bisimilar}, written $s\approx s'$,  
if there exists a bisimulation $B$ s.t.~$\tup{s,s'}\in B$.
Similarly to simulations, it can be proven that the 
bisimilarity relation $\approx$ is the largest bismulation.
Further, it is an equivalence relation.
Finally, $\P$ and $\P'$ are said to be 
{\em bisimilar}, written $\P \approx \P'$, if 
$s_0\approx s_0'$.
%
%
%

Since, as shown in Proposition~\ref{prop:iso-inst-eq2}, the
satisfaction of FO-sentences is invariant under isomorphisms, we can
now extend the usual bisimulation result from the propositional case
to that of SA-FO-CTL.
We begin by showing a result on bisimilar runs.


%
\begin{proposition}\label{prop:run}
Consider two \aqis\ $\P$ and $\P'$
such that $\P \approx \P'$, $s \approx s'$, for some
$s \in \S, s' \in \S'$, and a run $r$ of $\P$ such that $r(0) = s$. 
Then there exists a run $r'$ of $\P'$ such that:
\begin{itemize}
	\item[\myi] $r'(0)=s'$; 
	\item[\myii] for all $i\geq 0$, $r(i) \approx r'(i)$.
\end{itemize}
\end{proposition}

\begin{proof}
We show by induction that such run $r'$ in $\P'$ exists.  For $i =
0$, let $r'(0) = s'$. Obviously, $r(0)\approx r'(0)$.  Now, assume, by
induction hypothesis, that 
$r(i)\approx r'(i)$.  Let $r(i) \to r(i+1)$. Since $r(i)\approx
r'(i)$, by Def.~\ref{def:SA-sim}, there exists $t' \in \S'$ such that
$r'(i) \to t'$ and $r(i+1)\approx t'$.
Let 
$r'(i+1)=t'$; hence we obtain $r(i+1)\approx r'(i+1)$.
By definition $r'$ is a run of $\P'$.
\end{proof}

This enables us to show that bisimilar \aqis\ preserve SA-FO-CTL
formulas. This is an extension of analogous results on propositional
CTL. 


\begin{lemma}\label{lm:bis-iff-eq2}
Consider the \aqis\ $\P$ and $\P'$
such that $\P \approx \P'$, $s \approx s'$, for some
$s \in \S, s' \in \S'$ and an SA-FO-CTL formula $\varphi$.
 Then, 
\begin{eqnarray*}
(\P, s) \models \varphi & \text{iff} & (\P',s') \models \varphi
\end{eqnarray*}
\end{lemma}

\begin{proof}
The proof is by induction on the structure of $\varphi$.  Observe
first that since $\varphi$ is sentence-atomic, its satisfaction does
not depend on assignments. We report the proof for the left-to-right
part of the implication; the converse can be shown similarly. 

The base case for an FO-sentence $\varphi$ follows from
Prop.~\ref{prop:iso-inst-eq2}.
The inductive cases for propositional connectives are straightforward.

For $\varphi \equiv AX \psi$, assume for contradiction that
$(\P,s) \models \varphi$ and $(\P',s') \not \models \varphi$.  Then,
there exists a run $r'$ s.t.~$r'(0)=s'$ and
$(\P',r'(1)) \not \models\psi$.  By Def.~\ref{def:SA-bisim-ext}
and~\ref{def:SA-sim} there exists a $t \in \S$ s.t.~$s \to t$ and
$t \approx r'(1)$. Further, by seriality of $\to$, $s \to t$ can be
extended to a run $r$ s.t.~$r(0)=s$ and $r(1) = t$.  By the induction
hypothesis we obtain that $(\P,r(1)) \not \models \psi$.  Hence,
$(\P,r(0)) \not \models AX \psi$, which is a contradiction.

	
For $\varphi\equiv E\psi U\phi$, 
let $r$ be a run with $r(0)=s$ such that there exists $k\geq 0$
such that $(\P,r(k))\models \phi$, and for every $j$, $0\leq j <k$
implies $(\P,r(j))\models \psi$.  By Prop.~\ref{prop:run} there exists
a run $r'$ s.t.~$r'(0)=s'$ and for all $i\geq 0$, $r'(i)\approx r(i)$.
By the induction hypothesis we have that for each $i \in \mathbb{N}$, 
$(\P,r(i))\models\psi$ iff
$(\P',r'(i))\models\psi$, and $(\P,r(i))\models\phi$
iff $(\P',r'(i))\models\phi$.  
Therefore, $r'$ is a run
s.t.~$r'(0)=s'$, $(\P',r'(k))\models\phi$, and for every
$j$, $0\leq j < k$ implies $(\P',r'(j))\models\psi$, i.e.,
$(\P',s')\models E\psi U\phi$.

For $\varphi\equiv A\psi U\phi$, assume for contradiction that
$(\P,s)\models\varphi$ and
$(\P',s')\not\models\varphi$.  
Then, there exists a run
$r'$ s.t.~$r'(0)=s'$ and for every $k\geq 0$, if 
$(\P',r'(k))\models\phi$, then there exists $j$ s.t.~$0\leq
j < k$ and $(\P',r'(j))\not\models\psi$.
By Prop.~\ref{prop:run} there exists a run $r$
s.t.~$r(0)=s$ and for all $i\geq 0$, $r(i)\approx r'(i)$.
Further, 
by the induction hypothesis we have that 
$(\P,r(i))\models\psi$ iff
$(\P',r'(i))\models\psi$ and $(\P,r(i))\models\phi$
iff $(\P',r'(i))\models\phi$. But then $r$ is s.t.~$r(0)=s$
and for every $k\geq 0$, if $(\P,r(k))\models\phi$, then
there exists $j$ s.t.~$0\leq j < k$ and
$(\P,r(j))\not\models\psi$. That is,
$(\P,s)\not\models A\psi U \phi$, which is a contradiction.

\end{proof}

By applying the result above to the case of $s=s_0$ and $s'=s'_0$, we
obtain the following. 
\begin{theorem} 
Consider the \aqis\ $\P$ and $\P'$ such that $\P \approx \P'$, and an
SA-FO-CTL formula $\varphi$.
We have
\begin{eqnarray*}
\P \models \varphi & \text{iff} & \P' \models \varphi
\end{eqnarray*}
\end{theorem}

In summary we have proved that bisimilar \aqis\ validate the same
SA-FO-CTL formulas.  In the next section we use this result to
reduce, under additional assumptions, the verification of an
infinite-state \aqis\ to that of a finite-state one.

\subsection{Finite Abstractions of Bisimilar AC-MAS}
\label{sec:finabs}

We now define a notion of finite abstraction for \aqis. We prove that
abstractions are bisimilar to the corresponding concrete model.  We
are particularly interested in finite abstraction; so we operate on a
special class of infinite models that we call \emph{bounded}.



%

\begin{definition}[Bounded \aqis]\label{def:abs-qis}
An \aqis\ $\P$ is {\em $b$-bounded}, for $b \in \mathbb{N}$, if for
        all $s \in \S$, $\card{\adom(s)} \leq b$.
\end{definition}
An \aqis\ is $b$-bounded if none of its reachable states contains more
than $b$ distinct elements. Observe that bounded \aqis\ may be defined
on infinite domains $U$. Furthermore, note that a $b$-bounded \aqis\ may
contain infinitely many states, all bounded by $b$. So $b$-bounded
systems are infinite-state in general.
%
Notice also that the value $b$ bounds only the number of distinct
individuals in a state, not the \emph{size} of the state itself,
i.e., the amount of memory required to accommodate 
the individuals. 
Indeed, the infinitely many elements of $U$ 
need an unbounded number of bits to be represented
(e.g., as finite strings),
so, even though each state is guaranteed 
to contain at most $b$ distinct elements,
nothing can be said about how large the actual 
space required by such elements is.
On the other hand, it should be clear that memory-bounded 
\aqis\ are finite-state (hence $b$-bounded, for some $b$).

Thus, seen as programs, $b$-bounded \aqis\ are in general
memory-unbounded. 
Therefore, for the purpose of verification, they cannot be trivially 
checked by generating all their executions --as it would be 
the case if they were memory-bounded-- like standard model checking
techniques typically do.
However, we will show later that any $b$-bounded infinite-state ACMAS
admits a finite abstraction which can be used to verify it.

%


We now introduce abstractions in a modular manner by first introducing
a set of abstract agents from a concrete \aqis.

\begin{definition}[Abstract agent] \label{def:ab-agents}
Let $A = \tup{\D, L, Act, Pr}$
be an agent defined on the interpretation
domain $U$. 
Given a 
set $U'$ of individuals, we define 
the {\em abstract
agent} $A' = \tup{\D', L', Act',Pr'}$ on $U'$
such that:
\begin{enumerate}
\item $\D_i' = \D_i$; 
\item $L_i'\subseteq \D_i'(U')$;
\item $Act_i' = Act_i$; 
\item\label{req:ab-agents-prtcl} $\alpha(\vec u') \in Pr_i'(l_i')$ iff there exist
	 $l_i \in L_i$ and $\alpha(\vec u)\in Pr_i(l_i)$
          s.t.~$l_i' \simeq l_i$, for some witness $\iota$, and 
          $\vec u'=\iota'(\vec u)$, for some bijection $\iota'$ 
          extending $\iota$ to $\vec u$.
\end{enumerate}

Given a set $Ag$ of agents defined on $U$, let $Ag'$ be the set of the corresponding abstract agents on $U'$.
\end{definition}

We remark that $A'$, as defined in
Definition~\ref{def:ab-agents}, is indeed an agent and complies with
Definition~\ref{def:ags}.
Notice that the protocol of $A'$ is defined on the basis of its
corresponding concrete agent $A$ and requires the existence of a
bijection between the elements in the local states and the action
parameters. Thus, in order for a ground action of $A$ to have a
counterpart in $A'$, the last requirement of
Definition~\ref{def:ab-agents} constrains $U'$ to contain a sufficient
number of distinct values.  As it will become apparent later, the size
of $U'$ determines how closely an abstract system can simulate its
concrete counterpart.


We can now formalize the notion of abstraction that we will use in
this section. 
\begin{definition}[Abstraction] \label{ref:abs1}
Let $\P$ be an \aqis\ over $Ag$ and $Ag'$ the set of agents obtained
as in Definition~\ref{def:ab-agents}, for some $U'$.  The \aqis\ $\P'$
defined over $Ag'$ is said to be an
{\em abstraction} of $\P$ iff:
\begin{itemize}
 \item $s'_{0} \simeq s_{0}$;
\item  $t' \in \tau'(s',\vec\alpha(\vec u'))$ for some $\vec\alpha(\vec u')\in Act(U')$ iff 
there exist $s,t \in \S$ and $\vec\alpha(\vec u) \in Act(U)$, such that 
  $t \in \tau(s,\vec\alpha(\vec u))$, $s\simeq s'$ and $t \simeq t'$ for some witness $\iota$, and $\vec{u}' = \iota'(\vec u)$ for some $\iota'$ extending $\iota$.
\end{itemize}
\end{definition}

Notice that abstractions have initial states isomorphic to their
concrete counterparts. The condition in Definition~\ref{ref:abs1}
means that whenever $s \simeq s'$ for some witness $\iota$, $\vec u'
= \iota(\vec u)$, $t
\in \tau(s, \alpha(\vec u))$ and $t' \in \tau(s', \alpha(\vec u'))$,
then $t \simeq t'$. This constraint means that action are
data-independent. So, for example, a copy action in the concrete model
has a corresponding copy action in the abstract model regardless of
the data that are copied. Crucially, this condition requires that the
domain $U'$ contains enough elements to simulate the concrete states
and action effects as the following result makes precise.
In what follows we take $N_{Ag} = N_{Ag'} = \sum_{A_i \in
Ag} \max_{\alpha{(\vec p)} \in Act_i} \{|\vec p|\}$, i.e., $N_{Ag}$ is
the sum of the maximum numbers of parameters contained in the action
types of each agent in $Ag$.

\begin{theorem} 
Consider a $b$-bounded \aqis\ $\P$ over an infinite interpretation
domain $U$, an SA-FO-CTLK formula $\varphi$, and a finite
interpretation domain $U'$ such that $C\subseteq U'$ and
$\card{U'} \geq b + \card{C} + N_{Ag}$. Any abstraction $\P'$ of $\P$
is bisimilar to $\P$.
\end{theorem}

\begin{proof}
Define a relation $R$ as $R=\set{\tup{s,s'}\in \S\times \S'\mid s\iso
s'}$. We show that $R$ is a bisimulation 
such that $\tup{s_0,s_0'}\in R$.
Observe first that 
$s_0'\simeq s_0$, so
$\tup{s_0,s_0'}\in R$.
Next, consider $s \in \S$ and $s'\in\S'$ such that $ s \iso s'$ 
(i.e., $\tup{s,s'}\in R$), 
and assume
that $s \to t$, for some $t \in \S$.
Then, there exists $\alpha(\vec u) \in Act(U)$
s.t.~$t \in \tau(s,\alpha(\vec u))$.
We show next that there exists $t' \in \S'$
 s.t.~$s' \to t'$ and $t\iso t'$.
To this end, observe that, since $\card{U'}\geq b + \card{C}$ and
$\card{\adom(t)} \leq b$, we can define an injective function
$f: \adom(t) \cup C \mapsto U'$ such that $ f(t) \iso t$. We take
$t'=f(t)$; it remains to prove that $s' \to t'$. 
By the condition on the cardinality of $U'$ we can extend $f$ to $\vec
u$ as well, and set $\vec u' = f(\vec u)$.  Then, by the definition of
$\P'$ we have that $t' \in \tau'(s',\alpha( \vec u' ))$.
Hence, $s' \to t'$.
So, $R$ is a simulation relation between $\P$ and $\P'$. Since
$R^{-1}$ can similarly be shown to be a simulation, it follows that
$\P$ and $\P'$ are bisimilar.
\end{proof}

By combining this result with Lemma~\ref{lm:bis-iff-eq2}, we can
easily derive the main result of this section.

\begin{theorem} \label{theor1.1}
If $\P$ is a $b$-bounded \aqis\ over an infinite interpretation domain $U$,
and $\P'$ an abstraction of $\P$ over a finite interpretation domain 
$U'$ such that $C\subseteq U'$ and $\card{U'} \geq b + \card{C} + N_{Ag}$,
then for every SA-FO-CTLK formula $\varphi$, we have that
\begin{eqnarray*}
\P \models \varphi & \text{iff} & \P' \models \varphi.
\end{eqnarray*}
\end{theorem}

This result states that we can reduce the verification of an
infinite \aqis\ to the verification of a finite one. Given the fact
that checking a finite \aqis\ is decidable, this is a noteworthy
result.  
Note, however, that we do not have a constructive definition for the
construction of an abstract \aqis\ $\mathcal{P}'$ from a
concrete \aqis\ $\mathcal{P}$. This is of no consequence though, as in
practice any concrete artifact-system will be defined by a program,
e.g., in the language GSM, as discussed in the
introduction. \label{page:nonconstructive} Of importance, instead, is
to be able to derive finite abstractions not just for
arbitrary \aqis\, but for those that are models of concrete
programs. We will do this in Section~\ref{sec:as} where we will use
the result above.

Observe that an abstract \aqis\ as in Definition~\ref{ref:abs1}
depends on the set $Ag'$ of abstract agents defined in
Definition~\ref{def:ab-agents}.  However, other abstract \aqis\
defined on different sets of agents, exist. This is a standard
outcome when defining modular abstractions, as the same system can be
obtained by considering different agent components.  





\section{Abstractions for FO-CTLK} \label{kr}

In the previous section we showed that syntactical restrictions on the
specification language lead to finite abstractions for bounded \aqis. 
A natural question that arises is whether the limitation to
sentence-atomic specifications can be removed. Doing so would enable
us to check any agent-based FO-CTLK specification not on an
infinite-state \aqis, but on its finite abstraction.

The key concept we identify in this section that enables us to achieve
the above is that of \emph{uniformity}. As we will see later
uniform \aqis\ are systems for which the behaviour does not depend on
the actual data present in the states. This means that the system
contains all possible transitions that are enabled according to
parametric action rules, thereby resulting in a rather ``full'' 
transition relation.
This notion 
corresponds to that of \emph{genericity} in databases~\cite{AbiteboulHV95}.
We use the term ``uniformity'' as we refer to transition systems and
not databases.

To achieve finite abstractions we proceed as follows. We first
introduce a notion of bisimulation stronger than the one discussed in the
previous section. In Subsection~\ref{obisim} we show that this new
bisimulation relation guarantees that uniform \aqis\ satisfy the same
formulas in FO-CTLK. We use this result to show that bounded, uniform systems
admit finite abstractions (Subsection~\ref{abstraction}). 

In the rest of the section we let $\P = \tup{\S, U, s_{0}, \tau}$ and
$\P' = \tup{\S', U', s'_{0}, \tau'}$ be two \aqis\ and assume, unless
stated differently, that $s = \tup{l_0, \ldots, l_n} \in \S$, 
and $s' = \tup{l'_0, \ldots, l'_n} \in \S'$.


\subsection{$\oplus$-Bisimulation} \label{obisim}

Plain bisimulations are known to be satisfaction preserving in a modal
propositional setting~\cite{BlackburndRV01}. In the following we
explore the conditions under which this applies to \aqis\ as well.
We begin by using a notion of bisimulation which is also based on
isomorphism, but it is stronger than the one discussed in
Section~\ref{icsoc} and later explore its properties in the context of
uniform \aqis.

\begin{definition}[$\oplus$-Simulation]\label{def:sim-ext}
A relation $R$ on $\S \times \S'$ is
a \emph{$\oplus$-simulation} 
if $\tup{s,s'}\in R$ implies:
\begin{enumerate} 
\item $s \iso s'$;
\item for every $t \in \S$, if $s \ra t$ then there exists
$t' \in \S'$ s.t.~$s' \ra t'$, $s \oplus t \iso s' \oplus t'$, and
$\tup{t,t'}\in R$; \label{one} 
\item for every $t \in \S$, for every $0 < i \leq n$, if  $s \sim_i t$ then there exists
$t' \in \S'$ s.t.~$t \sim_i t'$, $s \oplus t \iso s' \oplus t'$, and
$\tup{t,t'}\in R$. \label{two}
\end{enumerate}
\end{definition}
Observe that Definition~\ref{def:sim-ext} differs from
Definition~\ref{def:SA-sim} not only by adding a condition for the
epistemic relation, but also by insisting that $s \oplus t \iso
s' \oplus t'$. This condition ensures that the $\oplus$-similar
transitions in \aqis\ have isomorphic disjoint unions. 
Two states $s\in \S$ and $s'\in \S'$ are said to be $\oplus$-\emph{similar},
iff there exists an $\oplus$-simulation $R$ s.t.~$\tup{s,s'}\in R$.
Note that all
$\oplus$-similar states are isomorphic as condition 2.~above ensures
that $t \iso t'$. We use the symbol $\preceq$ both for similarity and
$\oplus$-similarity, as the context will disambiguate. 
Also
$\oplus$-similarity can be shown to be the largest $\oplus$-simulation,
reflexive, and transitive.
Further, we say that $\P'$ $\oplus$-\emph{simulates} $\P$ if $s_0\preceq s_0'$.


$\oplus$-simulations can naturally be extended to
$\oplus$-bisimulations.

\begin{definition}[$\oplus$-Bisimulation]\label{def:bisim-ext}
A relation $B$ on $\S \times \S'$ is
a \emph{$\oplus$-bisimulation} 
iff both
$B $ and ${B}^{-1}=\set{\tup{s',s}\mid \tup{s,s'}\in B}$ are
$\oplus$-simulations. 
\end{definition}
Two states $s\in \S$ and $s'\in \S'$ are said to be $\oplus$-bisimilar
iff there exists an $\oplus$-bisimulation $B$ such that $\tup{s,s'}\in B$. 
Also for bisimilarity and $\oplus$-bisimilarity, we use the same  
symbol, $\approx$, and can prove that $\approx$ is the largest 
$\oplus$-bisimulation, and an equivalence relation.
We say that $\P$ and $\P'$ are {\em $\oplus$-bisimilar},
written $\P\approx \P'$ iff so are $s_0$ and $s'_0$.


While we observed in the previous section that bisimilar, hence
isomorphic, states in bisimilar systems preserve sentence atomic
formulas, it is instructive to note that this is not the case when
full FO-CTLK formulas are considered.
\begin{figure}
\centering

\begin{tikzpicture}[auto,node distance=1.5cm,->,>=stealth',shorten
>=1pt,semithick]

\tikzstyle{every state}=[fill=white,draw=black,text=black,minimum size=0pt,rectangle]
\tikzstyle{every initial by arrow}=[initial text=]

\begin{scope}
\node[initial left,state] 	(s0)   	{$1$};
\node[state]    			(s1) 	[right of=s0]	{$2$};
\node[state]    			(s2) 	[right of=s1]	{$3$};
\node[state]    			(s3) 	[right of=s2]	{$4$};
\node[state]    			(s4) 	[right of=s3]	{$5$};
\node    				(s5) 	[right of=s4]	{};
\node    				(s6) 	[right of=s5]	{};
\node[draw=none,fill=none,yshift=.8cm] 	(label) at (s0)	{$\P$};
\path 	(s0) 	edge (s1);
\path 	(s1) 	 edge (s2);
\path 	(s2) 	 edge (s3);
\path 	(s3) 	 edge (s4);
\path[dashed] 	(s4) 	 edge (s5);
\end{scope}

\begin{scope}[yshift=-2cm]
\node[initial left,state] 	(s0) 				{$1$};
\node[state]    			(s1) 	[right of=s0]	{$2$};

\node[draw=none,fill=none,yshift=.8cm] 	(label) at (s0)	{$\P'$};
\path 	(s0) 	edge[bend left] (s1);
\path 	(s1) 	edge[bend left] (s0);
\end{scope}

\end{tikzpicture}
\caption{$\oplus$-Bisimilar \aqis\ not satisfying the same FO-CTLK formulas.\label{fig:bisim-systems}}
\end{figure}
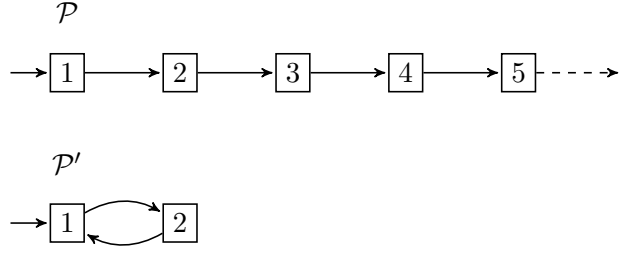

{\bf Example.} Consider Figure~\ref{fig:bisim-systems}, where
$\Const=\emptyset$ and $\P$ and $\P'$ are given as follows. 
For the number $n$ of agents equal to 1, we define
$\D=\D'=\set{P/1}$ and
$U=\mathbb{N}$; $s_0(P)=s_0'(P)=\set{1}$; $\tau=\set{\tup{s,s'}\mid
s(P)=\set{i}, s'(P)=\set{i+1}}$; $\tau'=\set{\tup{s,s'}\mid
s(P)=\set{i}, s'(P)=\set{(i+1)\mod 2}}$.
Notice that $\S\subseteq \D(\mathbb{N})$ and $\S'\subseteq \D(\mathbb{N})$.
Clearly we have that
$\P\approx \P'$. Now, consider the constant-free FO-CTLK formula
$\varphi=AG (\forall x (P(x) \ra AX AG \neg P(x)))$. It can be easily
seen that $\P\models\varphi$ while $\P'\not\models\varphi$.

The above shows that $\oplus$-bisimilarity is not a sufficient condition to
guarantee preservation of the satisfaction of FO-CTLK
formulas. Intuitively, this is a consequence of the fact that
$\oplus$-bisimilar \aqis\ do not preserve value associations along runs. For
instance, the value $1$ in $\P'$ is infinitely many times associated
with the odd values occurring in $\P$. By quantifying across states we
are able to express this fact and are therefore able to distinguish
the two structures. This is a difficulty as, intuitively, we would
like to use $\oplus$-bisimulations to demonstrate the existence of finite
abstractions. Indeed, as we will show later, this happens for the
class of \emph{uniform} \aqis, defined below. 
\begin{definition}[Uniformity]\label{def:unif-aqis}
An \aqis\ $\P$ 
is said to be \emph{uniform} iff for every $s,t,s' \in \S$, $t' \in
\D(U)$, 
\begin{enumerate}
	\item\label{req:unif-1} 
if $t \in \tau(s,\vec\alpha(\vec{u}))$ 
and $s \oplus t \iso s' \oplus t'$ for some witness $\iota$, then for
every constant-preserving bijection $\iota'$ that extends $\iota$ to
$\vec u$, we have that 
$t' \in \tau(s',\vec\alpha(\iota'(\vec{u})))$; \item\label{req:unif-ind}
if $s \sim_i t$ and $s \oplus t \iso s'\oplus t'$, then $s'\sim_i t'$.
\end{enumerate}
\end{definition}

This definition captures the idea that
actions take into account and operate only on 
the relational structure of states and action parameters, 
irrespectively of the actual data they contain 
(apart from a finite set of constants).
Intuitively, it says that if 
$t$ can be obtained by 
executing $\alpha(\vec u)$ in $s$, and we 
replace in $s$, $\vec u$ and $t$, 
the same element $v$ with
 $v'$, obtaining, say,
 $s'$, $\vec u'$ and $t'$, then $t'$ can be 
 obtained by executing $\alpha(\vec u')$ in $s'$.
In terms of the underlying Kripke structures this means that the
systems are ``full'' up to $\oplus$, i.e., in all uniform \aqis\ the
points $t'$ identified above are indeed part of the system and
reachable from $s'$. 
A similar condition is required on the epistemic relation. 
A useful property of uniform systems is the fact that the latter
requirement is implied by the former, as shown by the following
result.

\begin{proposition}\label{prop1}
If an \aqis\ $\P$ satisfies req.~\ref{req:unif-1} in Def.~\ref{def:unif-aqis}
and $\adom(s_0)\subseteq C$, then req.~\ref{req:unif-ind} is also satisfied.
\end{proposition}

\begin{proof}
	If $s \oplus t\iso s'\oplus t'$, then there is a 
	witness $\iota : \adom(s) \cup  \adom(t) \cup C \mapsto \adom(s') \cup  \adom(t') \cup C$ 
	that is the identity on $C$ (hence on $\adom(s_0))$.
	Assume $s \sim_i t$, thus $l_i(s) = l_i(t)$, 
	and $l_i(s')=\iota(l_i(s))=\iota(l_i(t))=l_i(t')$.
	Notice that this does not guarantee that $s' \sim_i t'$, 
	as we need to prove that $t' \in \S$. This can be done by
	showing that $t'$ is reachable from $s_0$. 
Since $t$ is reachable from $s_0$, there exists a run $s_0 \to s_1 \to \ldots \to s_k$ s.t.~$s_k=t$. 
Extend now $\iota$ to a total and injective function 
$\iota':\adom(s_0)\cup\cdots\cup\adom(s_k)\cup C\mapsto U$.
This can always be done because
$\card{U}\geq\card{\adom(s_0)\cup\cdots\cup\adom(s_k)\cup C}$.
Now consider the sequence $\iota'(s_0), \iota'(s_1), \ldots , \iota'(s_k)$.
Since $\adom(s_0) \subseteq C$ then $\iota(s_0)=s_0$ and, because $\iota'$ extends $\iota$,
we have that $\iota'(s_0) = \iota(s_0) = s_0$. 
Further, $\iota'(s_k) = \iota(t) = t'$.
By repeated applications of req.~\ref{req:unif-1} 
we can show that
$\iota'(s_{m+1}) \in \tau(\iota'(s_m),\vec\alpha(\iota'(\vec{u})))$
whenever
$s_{m+1} \in \tau(s_m,\vec\alpha(\vec{u}))$, for $m < k$. Hence,
the sequence is actually a run 
from $s_0$ to $t'$. Thus, $t' \in \S$, and $s'\sim_i t'$.
\end{proof}
Thus, as long as $\adom(s_0)\subseteq\Const$, 
to check whether an \aqis\ is uniform, it is sufficient to take
into account  only the transition function. 

A further distinctive feature of uniform systems is that all
isomorphic states are $\oplus$-bisimilar.
 

%
\begin{proposition}\label{prop:uni-bisim}
If an \aqis\ $\P$ is uniform, then for every $s, s' \in\S$,
$s \iso s'$ implies $s \approx s'$.
\end{proposition}

\begin{proof}
We prove that $B=\set{\tup{s,s'}\in\S\times \S\mid s\iso s'}$ is a
$\oplus$-bisimulation.  Observe that since $\iso$ is an equivalence
relation, so is $B$.  Thus $B$ is symmetric and $B=B^{-1}$.
Therefore, proving that $B$ is a $\oplus$-simulation proves also that
$B^{-1}$ is a $\oplus$-simulation; hence, that $B$ is a
$\oplus$-bisimulation.
To this end, let $\tup{s,s'}\in B$, and
assume $s \to t$ for some $t \in \S$. Then,
$t \in \tau(s,\alpha(\vec{u}))$ for some $\alpha(\vec{u}) \in Act(U)$.
Consider a witness $\iota$ for $s \iso s'$.  By cardinality
considerations $\iota$ can be extended to a total and injective
function $\iota' : \adom(s) \cup \adom(t) \cup \{ \vec u \}\cup C \mapsto U$.  Consider
$\iota'(t) = t'$; it follows that $\iota'$ is a witness for $s \oplus
t \iso s' \oplus t'$. Since $\P$ is uniform,
$t' \in \tau(s',\alpha(\iota'(\vec{u})))$, that is, $s' \to t'$.
Moreover, $\iota'$ is a witness for $t \iso t'$, thus $\tup{t, t'}\in
B$.
Next assume that $\tup{s,s'}\in B$ and
	$s \sim_i t$, for some $t \in \S$. 
	By reasoning as above we can
	find
a witness $\iota$ for $s \iso s'$, and an extension $\iota'$ of
	$\iota$ 
	s.t.~$t' = \iota'(t)$ and $\iota'$ is a witness for $s \oplus
	t \iso s' \oplus t'$. Since $\P$ is uniform, $s' \sim_i t'$
        and $\tup{t,
	t'}\in B$.
\end{proof}

This result intuitively means that submodels generated by
isomorphic states are $\oplus$-bisimilar.

Next we prove some partial results, which will be useful in proving
our main preservation theorem.  The first two results guarantee that
under appropriate cardinality constraints the $\oplus$-bisimulation
preserves the equivalence of assignments w.r.t.~a given FO-CTLK
formula.
\begin{lemma}\label{lem:nxt-uniform}
Consider two $\oplus$-bisimilar and uniform \aqis\ $\P$ and $\P'$, two
$\oplus$-bisimilar states $s \in \S$ and $s' \in \S'$,
and an FO-CTLK formula $\varphi$.
For every assignments $\sigma$ and $\sigma'$
equivalent for $\varphi$ w.r.t.~$s$ and $s'$, we have that:
\begin{enumerate}
\item\label{req:nxt-uniform-1} for every $t \in \S$ s.t.~$s \ra t$, if  
	$\card{U'} \geq\ |\adom(s)\cup\adom(t)\cup
	C \cup \sigma(\free(\varphi))|$, then there exists
	$t' \in \S'$ s.t.~$s' \ra t'$, $t \approx t'$, and $\sigma$
	and $\sigma'$ are equivalent for $\varphi$ w.r.t.~$t$ and
	$t'$.
\item\label{req:nxt-uniform-2} for every $t \in \S$ s.t.~$s \sim_i t$, if
	$\card{U'}\geq |\adom(s)\cup\adom(t)\cup
	C \cup \sigma(\free(\varphi))|$, then there exists
	$t' \in \S'$ s.t.~$s' \sim_i t'$, $t \approx t'$, and $\sigma$
	and $\sigma'$ are equivalent for $\varphi$ w.r.t.~$t$ and
	$t'$.
\end{enumerate}
\end{lemma}
\textbf{Proof.}
To prove~(\ref{req:nxt-uniform-1}), let $\gamma$ be a bijection
witnessing that $\sigma$ and $\sigma'$ are equivalent for $\varphi$
w.r.t.~$s$ and $s'$.
Suppose that $s \ra t$. 
Since $s \approx s'$, by definition of $\oplus$-bisimulation 
there exists $t'' \in \S'$ s.t.~$s' \ra t''$, 
$s \oplus t \iso s' \oplus t''$, and $t \approx t''$.
Now, define $Dom_j \doteq \adom(s)\cup\adom(t)\cup
C$, and partition it into:
\begin{itemize}
	\item $Dom_\gamma \doteq \adom(s)\cup C\cup(\adom(t)\cap \sigma(\free(\varphi))$;
	\item $Dom_{\iota'} \doteq \adom(t)\setminus Dom_\gamma$.
\end{itemize}

Let $\iota':Dom_{\iota'} \mapsto U' \setminus \im(\gamma)$ be an
invertible (total) function. 
Observe that $\card{\im(\gamma)}
= \card{\adom(s') \cup C \cup \sigma'(\free(\varphi))}
= |\adom(s)\cup C\cup \sigma(\free(\varphi))|$, 
thus from the fact that  $\card{U'} \geq \card{\adom(s) \cup \adom(t) \cup C\cup \sigma(\free(\varphi))}$
we have $\card{U' \setminus\im(\gamma)} \geq \card{Dom(\iota')}$, 
which guarantees the existence of $\iota'$.

Next, define $j: Dom_j\mapsto U'$ as follows:
\[
	j(u) =\left \{
		\begin{array}{l}
			\gamma(u) \mbox{, if } u\in Dom_\gamma\\
			\iota'(u) \mbox{, if } u\in Dom_{\iota'}
		\end{array}\right.	
\]

Obviously, $j$ is invertible. Thus, $j$ is a witness
for $s \oplus t \iso s' \oplus t'$, where $t' = j(t)$.
Since $s \oplus t \iso s' \oplus t''$ and $\iso$ is an equivalence
relation,
$s' \oplus t' \iso s' \oplus t''$.
Thus, $s' \ra t'$, as $\P'$ is uniform.
Moreover, 
$\sigma$ and $\sigma'$ are equivalent for $\varphi$ 
w.r.t.~$t$ and $t'$, by construction of $t'$. 
To check that $t \approx t'$, observe that, 
since $t' \iso  t''$ and $\P'$ is uniform, 
by Prop.~\ref{prop:uni-bisim} it follows that  $t' \approx  t''$.
Thus, since $t \approx  t''$ and $\approx$ 
is transitive, we obtain that $t \approx  t'$.
The proof for~(\ref{req:nxt-uniform-2}) has an analogous structure and
is omitted.
\qed

It can be proven that this result is tight, i.e., that if the 
cardinality requirement is violated, 
there exist cases where assignment equivalence is not 
preserved along temporal or epistemic transitions.

Lemma~\ref{lem:nxt-uniform} easily generalizes to t.e. runs.
\begin{lemma}\label{lem:run-uniform}
Consider two $\oplus$-bisimilar and uniform \aqis\ $\P$ and $\P'$, 
two $\oplus$-bisimilar states $s \in \S$ and $s' \in \S'$,
an FO-CTLK formula $\varphi$,
and two assignments $\sigma$ and $\sigma'$  
equivalent for $\varphi$ w.r.t.~$s$ and $s'$.
For every t.e.~run $r$ of $\P$, if $r(0) = s$ and 
for all $i \geq 0$, 
$\card{U'}\geq |\adom(r(i))\cup\adom(r(i+1))\cup C\cup \sigma(\free(\varphi))|$,
then there exists a t.e.~run $r'$ of $\P'$ s.t.~for all $i\geq 0$: 
\begin{itemize}
	\item[\myi] $r'(0)=s'$;
	\item[\myii] $r(i) \approx r'(i)$;
	\item[\myiii] $\sigma$ and $\sigma'$ are equivalent for $\varphi$ w.r.t.~$r(i)$ and $r'(i)$.
	\item[\myiv] for every $i\geq 0$, 
		if $r(i) \ra r(i+1)$ then $r'(i) \ra r'(i+1)$, and
		if $r(i) \sim_j r(i+1)$, for some $j$,  then $r'(i) \sim_j r'(i+1)$.
\end{itemize}
\end{lemma}

\begin{proof}
Let $r$ be a t.e.~run s.t.~$\card{U'}\geq
|\adom(r(i))\cup\adom(r(i+1))\cup C\cup \sigma(\free(\varphi))|$ for
all $i\geq 0$.  We inductively build $r'$ and show that the conditions
above are satisfied.  For $i = 0$, let $r'(0) = s'$. By hypothesis,
$r$ is s.t.~$\card{U'} \geq \card{\adom(r(0))\cup\adom(r(1))\cup
C \cup \sigma(\free(\varphi))}$.  Thus, since $r(0) \leadsto r(1)$, by
Lemma~\ref{lem:nxt-uniform} there exists $t' \in \S'$ s.t.~$r'(0) \leadsto
t'$, $r(1)\approx t'$, and $\sigma$ and $\sigma'$ are equivalent for
$\varphi$ w.r.t.~$r(1)$ and $t'$.  Let $r'(1)= t'$.
Lemma~\ref{lem:nxt-uniform} guarantees that the transitions
$r'(0) \leadsto t'$ and $r(0) \leadsto r(1)$ can be chosen so that
they are either both temporal or both epistemic with the same index.
	
The case for $i > 0$ is similar.  Assume that $r(i)\approx r'(i)$ and
$\sigma$ and $\sigma'$ are equivalent for $\varphi$ w.r.t.~$r(i)$ and
$r'(i)$. Since $r(i)\leadsto r(i+1)$ and $\card{U'}\geq
|\adom(r(i))\cup\adom(r(i+1))\cup C\cup \sigma(\free(\varphi))|$, by
Lemma~\ref{lem:nxt-uniform} there exists $t' \in \S'$
s.t.~$r'(i)\leadsto t'$, $\sigma$ and $\sigma'$ are equivalent for
$\varphi$ w.r.t.~$r(i+1)$ and $t'$, and $r(i+1)\approx t'$. Let
$r'(i+1) = t'$. It is clear that $r'$ is a t.e.~run in $\P'$, and
that, by Lemma~\ref{lem:nxt-uniform}, the transitions of $r'$ can be
chosen so as to fulfill requirement \myiv.
\end{proof}

We can now prove the following result, which states that FO-CTLK
formulas cannot distinguish $\oplus$-bisimilar and uniform \aqis. This is in
marked contrast with the earlier example in this section which
operated on $\oplus$-bisimilar but non-uniform \aqis.
\begin{theorem}\label{th:bis-iff-eq}
Consider two $\oplus$-bisimilar and uniform \aqis\ $\P$ and $\P'$, 
two $\oplus$-bisimilar states $s \in \S$ and $s' \in \S'$, 
an FO-CTLK formula $\varphi$,
and two assignments $\sigma$ and $\sigma'$ 
equivalent for $\varphi$ w.r.t.~$s$ and $s'$. 

If
\begin{enumerate}
	\item\label{req:run-card-v} 
for every t.e.~run $r$ s.t.~$r(0)=s$, 
		for all $k\geq 0$ we have  
		$\card{U'} \geq |\adom(r(k))\cup\adom(r(k+1))\cup C\cup    
 \sigma(\free(\varphi))| +
		|\vars(\varphi)\setminus\free(\varphi)|$; and
	\item\label{req:run-card-v2} 
for every t.e.~run $r'$ s.t.~$r'(0)=s'$, 
		for all $k\geq 0$ we have  
		$\card{U} \geq |\adom(r'(k))\cup\adom(r'(k+1))\cup C \cup \sigma'(\free(\varphi))| +
		\card{\vars(\varphi)\setminus\free(\varphi)}$;
\end{enumerate}
then 
\begin{eqnarray*}
(\P, s ,\sigma) \models \varphi & \text{iff} & (\P', s', \sigma') \models \varphi.
\end{eqnarray*}
\end{theorem}

\begin{proof}
The proof is by induction on the structure of $\varphi$. We prove
that if $(\P, s ,\sigma) \models \varphi$ then $(\P',
s', \sigma') \models \varphi$.  The other direction can be proved
analogously. 
The base case for atomic formulas follows from
Prop.~\ref{prop:iso-inst-eq}.
The inductive cases for propositional connectives are straightforward.

For $\varphi\equiv\forall x\psi$, assume that $x\in\free(\psi)$
(otherwise consider $\psi$, and the corresponding case), and no variable is quantified more than
once (otherwise rename the other variables).  Let $\gamma$ be a bijection
witnessing that $\sigma$ and $\sigma'$ are equivalent for $\varphi$
w.r.t.~$s$ and $s'$.
For $u\in\adom(s)$, consider the assignment $\sigma{x \choose u}$.  By
definition, $\gamma(u)\in\adom(s')$, and $\sigma'{x \choose \gamma(u)}$
is well-defined.  Note that 
$\free(\psi)=\free(\varphi)\cup\set{x}$; so
$\sigma{x \choose u}$ and $\sigma'{x \choose \gamma(u)}$ are
equivalent for $\psi$ w.r.t.~$s$ and $s'$.
Moreover, $\card{\sigma{x \choose
u}(\free(\psi))}\leq\card{\sigma(\free(\varphi))}+1$, as $u$ may not
occur in $\sigma(\free(\varphi))$.  The same considerations apply to $\sigma'$.
Further, $\card{\vars(\psi)\setminus{\free(\psi)}}
= \card{\vars(\varphi)\setminus{\free(\varphi)}}-1$, as
$\vars(\psi)=\vars(\varphi)$, $\free(\psi)=\free(\varphi)\cup\set{x}$,
and $x\notin\free(\varphi)$.  
Thus, both hypotheses \ref{req:run-card-v}.~and 
\ref{req:run-card-v2}.~remain satisfied if we 
        replace $\varphi$ with $\psi$, $\sigma$ with
        $\sigma{x \choose u}$, and $\sigma'$ with
        $\sigma'{x \choose \gamma(u)}$.
Therefore, by the induction hypothesis,
        if $(\P,s,\sigma{x\choose u})\models\psi$ then
        $(\P',s', \sigma'{x\choose \gamma(u)})\models\psi$.
Since $u\in\adom(s)$ is generic and $\gamma$ is a bijection, 
the result follows.

For $\varphi\equiv AX\psi$, assume by contradiction that
$(\P,s, \sigma)\models\varphi$ but
$(\P',s', \sigma')\not\models\varphi$.  Then, there exists a run $r'$
s.t.~$r'(0)=s'$ and $(\P',r'(1), \sigma')\not\models\psi$.  By
Lemma~\ref{lem:run-uniform}, which applies as
$\card{\vars(\varphi)\setminus{\free(\varphi)}}\geq 0$, there exists a
run $r$ s.t.~$r(0)=s$, for all $i\geq 0$, $r(i)\approx r'(i)$ and
$\sigma$ and $\sigma'$ are equivalent for $\psi$ w.r.t.~$r(i)$ and
$r'(i)$.  Since $r$ is a run s.t.~$r(0)=s$, it satisfies
hypothesis~\ref{req:run-card-v}.  Moreover, the same hypothesis is
necessarily satisfied by all the t.e.~runs $r''$ s.t., , for some
$i\geq 0$, $r''(0)=r(i)$ (otherwise, the t.e.~run $r(0)\cdots r(i)
r''(1) r''(2)\cdots$ would not satisfy the hypothesis); the same
considerations apply w.r.t~hypothesis~\ref{req:run-card-v2} and for
all the t.e.~runs $r'''$ s.t.~$r'''(0)=r'(i)$, for some $i\geq 0$.  In
particular, these hold for $i=1$.
Thus, we can inductively apply the Lemma, by replacing $s$ with
$r(1)$, $s'$ with $r'(1)$, and $\varphi$ with $\psi$ (observe that
$\vars(\varphi)=\vars(\psi)$ and $\free(\varphi)=\free(\psi))$.  But
then we obtain $(\P,r(1),\sigma)\not\models\psi$, thus
$(\P,r(0), \sigma)\not\models AX\psi$. This is a contradiction.
	
For $\varphi\equiv E\psi U\phi$, assume that the only variables common
to $\psi$ and $\phi$ occur free in both formulas (otherwise
rename the quantified variables).
Let $r$ be a run s.t.~$r(0)=s$, and there exists $k\geq 0$
s.t.~$(\P,r(k), \sigma)\models \phi$, and
$(\P,r(j), \sigma)\models \psi$ for $0\leq j <k$.  By
Lemma~\ref{lem:run-uniform} there exists a run $r'$ s.t.~$r'(0)=s'$,
and for all $i\geq 0$, $r'(i)\approx r(i)$, and $\sigma$ and $\sigma'$
are equivalent for $\varphi$ w.r.t.~$r'(i)$ and $r(i)$.
From each bijection $\gamma_i$ witnessing that $\sigma$ and $\sigma'$
are equivalent for $\varphi$ w.r.t.~$r'(i)$ and $r(i)$, define the
bijections $\gamma_{i,\psi}=\gamma_i|_{\adom(r(i))\cup
C\cup \sigma(\free(\psi))}$ and
$\gamma_{i,\phi}=\gamma_i|_{\adom(r(i))\cup
C\cup \sigma(\free(\phi))}$.  Since
$\free(\psi)\subseteq\free(\varphi), \free(\phi)\subseteq\free(\varphi)$,
it can be seen that $\gamma_{i,\psi}$ and $\gamma_{i,\phi}$ witness
that $\sigma$ and $\sigma'$ are equivalent for respectively $\psi$ and
$\phi$ w.r.t.~$r'(i)$ and $r(i)$.
By the same argument used for the $AX$ case above, 
hypothesis~\ref{req:run-card-v} holds for all the t.e.~runs $r''$ s.t.~$r''(0)=r(i)$,
for some $i \geq 0$, and hypothesis~\ref{req:run-card-v2} holds for all the
 t.e.~runs $r'''$ s.t.~$r'''(0)=r'(i)$.
Now observe that
$\card{\sigma(\free(\phi))},\card{\sigma(\free(\psi))}\leq \card{\sigma(\free(\varphi))}$.
Moreover, by the assumption on the common variables of $\psi$ and $\phi$,
$(\vars(\varphi)\setminus\free(\varphi))=(\vars(\psi)\setminus\free(\psi))\uplus(\vars(\phi)\setminus\free(\phi))$,
thus
$\card{\vars(\varphi)\setminus\free(\varphi)}=\card{(\vars(\psi)\setminus\free(\psi)}+\card{(\vars(\phi)\setminus\free(\phi)}$, 
hence
$\card{(\vars(\psi)\setminus\free(\psi)},\card{(\vars(\phi)\setminus\free(\phi)}\leq \card{\vars(\varphi)\setminus\free(\varphi)}$.
Therefore hypotheses~\ref{req:run-card-v} and \ref{req:run-card-v2}
hold also with $\varphi$ uniformly replaced by $\psi$ or $\phi$.
Then, the induction hypothesis applies for each $i$, by replacing $s$
with $r(i)$, $s'$ with $r'(i)$, and $\varphi$ with either $\psi$ or
$\phi$.
Thus, for each $i$, $(\P,r(i), \sigma)\models\psi$ iff
$(\P',r'(i), \sigma')\models\psi$, and $(\P,r(i), \sigma)\models\phi$
iff $(\P',r'(i), \sigma')\models\phi$.  
Therefore, $r'$ is a run
s.t.~$r'(0)=s'$, $(\P',r'(k), \sigma')\models\phi$, and for every
$j$, $0\leq j < k$ implies $(\P',r'(j), \sigma')\models\psi$, i.e.,
$(\P',s', \sigma')\models E\psi U\phi$.

For $\varphi\equiv A\psi U\phi$, assume by contradiction that
$(\P,s,\sigma)\models\varphi$ but
$(\P',s', \sigma')\not\models\varphi$.  Then, there exists a run $r'$
s.t.~$r'(0)=s'$ and for every $k\geq 0$, either
$(\P',r'(k), \sigma')\not\models\phi$ or there exists $j$
s.t.~$0\leq j < k$ and $(\P',r'(j), \sigma')\not\models\psi$.
By Lemma~\ref{lem:run-uniform} there exists a run $r$
s.t.~$r(0)=s$, and for all $i\geq 0$, $r(i)\approx r'(i)$ and $\sigma$ and
$\sigma'$ are equivalent for $\varphi$ w.r.t.~$r(i)$ and $r'(i)$.
Similarly to the case of $E\psi U\phi$, 
it can be shown that 
$\sigma$ and $\sigma'$ are equivalent for $\psi$ and $\phi$
w.r.t.~$r(i)$ and $r'(i)$, for all $i\geq 0$.
Further, assuming w.l.o.g.~that all variables common to $\psi$ and
$\phi$ occur free in both formulas, it can be shown, as in the
case of $E\psi U\phi$, that the induction hypothesis holds on every
pair of runs obtained as suffixes of $r$ and $r'$, starting from their
i-th state, for every $i\geq 0$.
Thus, $(\P,r(i), \sigma)\models\psi$ iff
$(\P',r'(i), \sigma')\models\psi$, and $(\P,r(i), \sigma)\models\phi$
iff $(\P',r'(i), \sigma')\models\phi$. But then $r$ is s.t.~$r(0)=s$
and for every $k\geq 0$, either $(\P,r(k), \sigma)\not\models\phi$ or
there exists $j$ s.t.~$0\leq j < k$ and
$(\P,r(j), \sigma)\not\models\psi$, that is,
$(\P,s, \sigma)\not\models A\psi U \phi$. This is a contradiction.
	
For $\varphi \equiv K_i \psi$, 
assume by contradiction that $(\P,s, \sigma)\models\varphi$ but
$(\P',s', \sigma')\not\models\varphi$.
Then, there exists $s''$ s.t.~$s' \sim_i s''$ and
$(\P',s'',\sigma') \not \models\psi$.
By Lemma~\ref{lem:run-uniform} there exists $s'''$ s.t.~$s''' \approx
s''$, $s \sim_i s'''$, and $\sigma$ and $\sigma'$ are equivalent for
$\psi$ w.r.t.~$s''$ and $s'''$.
Thus, by an argument analogous to that used for the case of $AX$,
we can apply the induction hypothesis,
obtaining $(\P,s''',\sigma)\not \models \psi$. 
But then $(\P,s, \sigma)\not \models K_i \psi$, which is a contradiction.

Finally, for $\varphi \equiv C \psi$, 
assume by contradiction that $(\P,s, \sigma)\models\varphi$ but
$(\P',s', \sigma')\not\models\varphi$.
Then, there exists an $s''$ s.t.~$s' \sim s''$ and
$(\P',s'',\sigma') \not \models\psi$.
Again by Lemma~\ref{lem:run-uniform} there exists $s'''$ s.t.~$s''' \approx
s''$, $s \sim s'''$, and $\sigma$ and $\sigma'$ are equivalent for
$\psi$ w.r.t.~$s''$ and $s'''$.
Thus, by an argument analogous to that used for the case of $K_i$,
we can apply the induction hypothesis,
obtaining $(\P,s''',\sigma)\not \models \psi$. 
But then $(\P,s, \sigma)\not \models C \psi$, which is a contradiction.
\end{proof}

We can now easily extend the above result to the model checking
problem for \aqis.
\begin{theorem}\label{th:bis-iff}
Consider two $\oplus$-bisimilar and uniform \aqis\ $\P$ and $\P'$, 
and an FO-CTLK formula $\varphi$.

If
\begin{enumerate}
	\item\label{req:run-card-gen} 
for all t.e.~runs $r$ s.t.~$r(0)=s_{0}$,
	and for all $k\geq 0$,
		$\card{U'} \geq |\adom(r(k))\cup\adom(r(k+1))\cup C| +
		|\vars(\varphi)|$, and
	\item\label{req:run-card-gen2} 
		for all t.e.~runs $r'$ s.t.~$r'(0)=s'_{0}$, 
		and for all $k\geq 0$,
		$\card{U} \geq |\adom(r'(k))\cup\adom(r'(k+1))\cup C| +
		|\vars(\varphi)|$
\end{enumerate}
then 
\begin{eqnarray*}
\P \models \varphi & \text{iff} & \P' \models \varphi.
\end{eqnarray*}
\end{theorem}

\begin{proof}
 Equivalently, we prove that if
$(\P,s_{0},\sigma) \not \models \varphi$ for some $\sigma$, then there
exists a $\sigma'$ such that $(\P',s'_{0},\sigma') \not \models \varphi$,
and viceversa.  To this end, observe that
hypotheses~\ref{req:run-card-gen}.~and \ref{req:run-card-gen2}.~imply,
respectively, hypotheses~\ref{req:run-card-v}.
and~\ref{req:run-card-v2}. of Theorem~\ref{th:bis-iff-eq}.
Further, notice that, by cardinality considerations, given the
assignment $\sigma: Var \mapsto U$, there exists an assignment
$\sigma': Var \mapsto U'$ s.t.~$\sigma$ and $\sigma'$ are equivalent
for $\varphi$ w.r.t.~$s_{0}$ and $s'_{0}$.
Thus, by applying Theorem~\ref{th:bis-iff-eq} we have that if there
exists an assignment $\sigma$
s.t.~$(\P,s_{0},\sigma)\not\models \varphi$, then there exists an
assignment $\sigma'$
s.t.~$(\P',s'_{0},\sigma')\not\models \varphi$.  
The converse can be proved analogously, as the hypotheses are symmetric.
\end{proof}

This result shows that uniform \aqis\ can in principle be verified by
model checking a $\oplus$-bisimilar one. Note that this applies to
infinite \aqis\ $\P$ as well. In this case the results above enable us
to show that the verification question can be posed on the
corresponding, possibly finite, $\P'$ as long as $U'$, as defined
above, is sufficiently large for $\P'$ to $\oplus$-bisimulate $\P$.
A noteworthy class of infinite systems for which these results prove 
particularly powerful is that of bounded \aqis, which, as discussed 
in the next subsection, always admit a finite abstraction.

\subsection{Finite Abstractions} \label{sec:finabs2}
We now  combine the notion of uniformity explored so far in this
section with the assumption on boundedness made in
Section~\ref{sec:finabs}. Our aim remains to identify conditions under
which the verification of an infinite \aqis\ can be reduced to the
verification of a finite one. Differently from
Section~\ref{sec:finabs} we here operate on the full FO-CTLK
specification language. The main result here is given by
Corollary~\ref{cor:preservation} which guarantees that, in the context
of bounded \aqis, uniformity is a sufficient condition for
$\oplus$-bisimilar finite abstractions to be satisfaction
preserving. 


In the following we assume that any \aqis\ $\P$ is such that
$\adom(s_0)\subseteq \Const$.  If this is not the case, $\Const$ can
be extended so as to include all the (finitely many) elements in
$\adom(s_0)$.
Further, 
we recall that  $N_{Ag}$ is the sum of the maximum numbers of parameters 
contained in the action types of each agent in $Ag$, i.e.,
$N_{Ag} = \sum_{A_i \in Ag} \max_{ \alpha(\vec x) \in Act_i} \{|\vec x|  \}$.

We start by formalizing the notion of $\oplus$-abstraction.
\begin{definition}[$\oplus$-Abstraction] \label{ref:abs}
Let $\P = \tup{ \S, U, s_0, \tau }$ be an \aqis\ over $Ag$, and $Ag'$
the set of abstract agents obtained as in
Definition~\ref{def:ab-agents}, for some domain $U'$.  The \aqis\ $\P'
= \tup{ \S', U', s'_0, \tau' }$ over $Ag'$ is said to be an {\em
  $\oplus$-abstraction} of $\P$ iff:
\begin{itemize}
\item $s'_{0} = s_{0}$;
\item $t' \in \tau'(s',\vec\alpha(\vec u'))$ iff there exist $s,
  t \in \S$ and $\vec\alpha(\vec u) \in Act(U)$, such that 
  $s \oplus t \simeq s' \oplus t'$, for some witness $\iota$,
$t \in \tau(s,\vec\alpha(\vec u))$,   
and $\vec u' = \iota'(\vec u)$ for some bijection $\iota'$ extending
$\iota$ to $\vec u$.
\end{itemize}
\end{definition}

Notice that $\P'$ is indeed an \aqis \ as it satisfies the relevant
conditions on protocols and transitions in
Definition~\ref{def:sys-db-ag}.  Indeed, if $t' \in
\tau'(s',\vec\alpha(\vec u'))$, then there exist $s, t \in \S$, and
$\vec\alpha(\vec u)$ such that~$t \in \tau(s,\vec\alpha(\vec u))$, $s
\oplus t \simeq s' \oplus t'$ for some witness $\iota$, and $\vec u =
\iota'(\vec u')$ for some bijection $\iota'$ extending $\iota$.  This
means that $\alpha_i(\vec{u}_i) \in Pr_i (l_i)$ for $i \leq n$.  By
definition of $Pr_i'$ we have that $\alpha_i(\vec{u}'_i) \in Pr'_i
(l'_i)$ for $i \leq n$.  Further, if $U'$ has finitely many elements,
then $\S'$ has finitely many states.
Observe that by varying $U'$ we obtain different $\oplus$-abstractions.

Next, we investigate the relationship between an \aqis\ and its
$\oplus$-abstractions.  A first useful result states that every finite
$\oplus$-abstraction is uniform, independently of the properties of
the \aqis\ they abstract.



\begin{lemma} \label{lem_unif}
Every $\oplus$-abstraction $\P'$ of an \aqis \ $\P$ is uniform.
\end{lemma}

\begin{proof}
Consider $s,t,s' \in \S'$, $t' \in \D(U')$, and $\vec\alpha(\vec u)
\in Act'(U')$ s.t.~$t \in \tau'(s,\vec\alpha(\vec{u}))$ and $s \oplus
t \iso s' \oplus t'$, for some witness $\zeta$. We need to show that
$\P'$ admits a transition from $s'$ to $t'$. 
Since $\P'$ is an $\oplus$-abstraction of $\P$, given the definition of $\tau'$,
there exist $s'',t''\in \S$ and $\vec\alpha(\vec u'') \in Act(U)$
s.t.~$t'' \in \tau(s'',\vec\alpha(\vec u''))$, 
$s'' \oplus t'' \simeq s \oplus t$, for some witness $\iota$, and 
$\vec u = \iota'(\vec u'')$, for some constant-preserving bijection $\iota'$
extending $\iota$ to $\vec u''$.
Consider $\vec u'\in U'^{\card{\vec u}}$ such that~$\vec u'=\zeta'(\vec u)$, for some constant-preserving bijection $\zeta'$ extending $\zeta$ to $\vec u$.
Obviously, the composition $\zeta'\circ\iota'$ is a constant-preserving bijection
such that $\vec u'=\zeta'(\iota'(\vec u''))$. Moreover, it can be easily restricted to a 
witness for $s''\oplus t''\simeq s'\oplus t'$.
But then, since $\P'$ is an $\oplus$-abstraction of $\P$, this implies that
$t'\in\tau'(s',\vec\alpha(\vec u'))$. Thus, $\P'$ is uniform.  
\end{proof}


The second result below guarantees that every $b$-bounded \aqis\ is
bisimilar to any of its $\oplus$-abstractions, provided these are
built over a sufficiently large interpretation domain.
\begin{lemma} \label{theor1}
Consider a uniform, $b$-bounded \aqis\ $\P$ over an infinite
interpretation domain $U$,
and an interpretation domain $U' $ such that $C\subseteq U'$.
If $\card{U'} \geq 2b + \card{C} + N_{Ag}$, then any
$\oplus$-abstraction
$\P'$ of $\P$ over $U'$ is bisimilar to $\P$.
\end{lemma}

\begin{proof}
Let $B=\set{\tup{s,s'}\in\S\times \S'\mid s\simeq s'}$. 
We prove that $B$ is a $\oplus$-bisimulation 
such that $\tup{s_0,s'_0}\in B$.
We start by proving that $B$ is a $\oplus$-simulation relation.
To this end, observe that since $s_0=s'_0$, 
then $s_0\simeq s'_0$, and $\tup{s_0,s_0'}\in B$.
Next, consider $\tup{s, s'}\in B$, 
thus $s\simeq s'$.
Assume that $s \to t$, for some $t \in \S$.
Then, there must exist $\vec\alpha(\vec u) \in Act(U)$
such that $t \in \tau(s,\vec\alpha(\vec u))$.
%
Moreover, since $\card{U'}\geq 2b + \card{C}+N_{Ag}$,
$\sum_{A_i \in Ag} \card{\vec u_i} \leq N_{Ag}$,
and $\card{\adom(s)\cup\adom(t)}\leq 2b$,
the witness $\iota$ for $s\iso  s'$ can 
be extended to 
$\bigcup_{A_i \in Ag} \vec u_i$ as a bijection $\iota'$.
Now let $t'= \iota'(t)$. By the way $\iota'$ has been defined,
it can be seen that
$s \oplus t \iso s'\oplus t'$.
Further, since $\P'$ is an $\oplus$-abstraction of $\P$, we have that
$t' \in \tau'(s',\vec\alpha(\vec u'))$ for $\vec u' = \iota'(\vec u)$, 
that is, $s' \to t'$ in $\P'$.
Therefore, there exists $t' \in\S'$
such that $s' \to t'$, $s \oplus t \iso s' \oplus t'$, and 
$\tup{t, t'}\in B$.
As regards the epistemic relation, assume $s \sim_i t$ for some $i\in
\set{1,\ldots, n}$ and $t \in \S$.  By definition of $\sim_i$,
$l_i(s)=l_i(t)$.  Since $\card{U'}\geq 2b+ \card{C}$, any witness
$\iota$ for $s\simeq s'$ can be extended to a witness $\iota'$ for
$s\oplus t\simeq s'\oplus t'$, where $t'=\iota'(t)$. Obviously,
$l_i(s')=l_i(t')$. Thus, to prove that $s'\sim_i t'$, we need to show
that $t'\in \S'$, i.e., that $t'$ is reachable in $\P'$ from
$s'_0=s_0$.  To this end, observe that since $t\in \S$, there exists a
purely temporal run $r$ such that $r(0)=s_0$ and $r(k) = t$, for
some $k\geq 0$.  Thus, there exist also $\vec\alpha^1(\vec
u^1)\ldots,\vec\alpha^k(\vec u^k)$ such that
$r(j+1)\in\tau(r(j),\vec\alpha^{j+1}(\vec u^{j+1}))$, for $0\leq j<
k$.  Since $\card{U'}\geq 2b + \card{C}$, we can define, for $0\leq
j< k$, a function $\iota_j$ that is a witness for $r(j)\oplus
r(j+1)\simeq \iota_j(r(j))\oplus \iota_j(r(j+1))$.
In particular, this can be done starting from $j=k-1$, defining
$\iota_{k-1}$ so that $\iota_{k-1}(r(k))=\iota_{k-1}(t)=t'$, and 
proceeding backward to $j=0$, guaranteeing that, for $0\leq j<k$,
$\iota_{j}(r(j+1))=\iota_{j+1}(r(j+1))$.
Observe that since $\adom(s_0)\subseteq C$, necessarily 
$i_0(r(0))=i_0(s_0)=s_0=s'_0$.
Moreover, as $\card{U'}\geq 2b + \card{C}+N_{Ag}$, each 
$\iota_j$ can be extended to a bijection $\iota'_j$, to the elements occurring 
in $\vec u^{j+1}$. 
Thus, given that $\P'$  is an $\oplus$-abstraction of $\P$, for $0\leq j < k$,
we have that $\iota_j'(r(j+1))\in \tau(\iota_j'(r(j)),\vec\alpha(\iota_j'(\vec u^{j+1})))$.
Hence, the sequence $\iota_0'(r(0))\to\cdots\to\iota_{k-1}'(r(k))$ is a run 
of $\P'$, and, since $t'=\iota_{k-1}'(r(k))$, $t'$ is reachable in $\P'$.
Therefore $s'\sim_i t'$.  Further, since $t\simeq t'$, by definition
of $B$, it is the case that $\tup{t,t'}\in B$, hence $B$ is a
$\oplus$-simulation.

To prove that $B^{-1}$ is a $\oplus$-simulation,
given $\tup{s,s'}\in B$ (thus $s\simeq s'$),
assume that $s' \to t'$, for some $t' \in \S'$.
Obviously, there exists $\vec\alpha(\vec u') \in Act(U')$
such that $t' \in \tau'(s',\vec\alpha(\vec u'))$.
Because $\P'$ is an $\oplus$-abstraction of $\P$,
there exist $s'',t''\in \S$ and $\vec\alpha(\vec u'')\in Act(U)$ 
such that $s''\oplus t''\simeq s' \oplus t' $, for
some witness $\iota$, and 
$t''\in\tau(s'',\alpha(\vec u''))$, with $\vec u''=\iota'(\vec u')$,
for some bijection $\iota'$ extending $\iota$ to $\vec u'$.
Observe that $s' \simeq s''$, thus, by transitivity of $\simeq$, 
we have $s\simeq s''$.  The fact that there exists $t\in \S$ such 
that $s\to t$ easily follows from the uniformity of $\P$.
Thus, since $t'\simeq t$, we have $\tup{t,t'}\in B$.
For the epistemic relation, assume $s'\sim_i t'$, for some $t'\in \S'$
and $0 < i\leq n$.  Let $\iota$ be a witness for $s'\simeq s$, and
let $\iota'$ be an extension of $\iota$ that is a witness for
$s'\oplus t'\simeq s\oplus t$.  For $t=\iota'(t')$, it can be seen
that $l_i(s)=l_i(t)$.  Observe that $t'\in\S'$.  Using an argument
essentially analogous to the one above, but exploiting the fact that
$\P$ is uniform, that $\P'$ is certainly $b$-bounded, and that
$\card{U}>2b+\card{C}+N_{Ag}$ as $U$ is infinite, we show that
$t\in\S$ by constructing a run $r$ of $\P$ such that $r(k)=t$, for
some $k\geq 0$. Then $s\sim_i t$. Further, since $t'\simeq t$, we have
$\tup{t,t'}\in B$. Therefore, $B^{-1}$ is a $\oplus$-simulation. So, 
$\P$ and $\P'$ are bisimilar.
\end{proof}

This result allows us to prove  our main abstraction theorem.
\begin{theorem} \label{theor2}
Consider a $b$-bounded and uniform \aqis\ $\P$ over an infinite 
interpretation domain $U$, 
an FO-CTLK formula $\varphi$, and an interpretation domain $U' $ such
that $C\subseteq U'$.  If $\card{U'} \geq 2b + \card{C} +
\max\set{\card{vars(\varphi)},N_{Ag}}$, then for any
$\oplus$-abstraction
$\P'$ of $\P$ over $U'$, we have that:
\begin{eqnarray*}
\P \models \varphi & \text{iff} & \P' \models \varphi.
\end{eqnarray*}
\end{theorem}
\begin{proof}
By Lemma~\ref{lem_unif}, $\P'$ is uniform.
Thus, by the hypothesis on the cardinalities of $U$ and $U'$,
Lemma~\ref{theor1} applies, so $\P$ and $\P'$ are bisimilar. 
Obviously, also $\P'$ is $b$-bounded.  Thus, since $\P$ and $\P'$ are
$b$-bounded, and by the cardinality hypothesis on $U$ and $U'$,
Theorem~\ref{th:bis-iff} applies. In particular, notice that for every
temporal-epistemic run $r$ s.t.~$r(0)=s_{0}$, and for all $k\geq 0$,
we have that $\card{U'} \geq |\adom(r(k))\cup\adom(r(k+1)) \cup C| +
|\vars(\varphi)|$, as $\card{\adom(r(k))}\leq b$, by $b$-boundedness.
Therefore, $\P \models \varphi$ iff $\P' \models\varphi$.
\end{proof}

Note that the theorem above does not require $U'$ to be infinite.  So,
by using a sufficient number of abstract values in $U'$, we can in
principle reduce the verification of an infinite, bounded, and uniform
\aqis\ to the verification of a finite one. The following corollary
to Theorem~\ref{theor2} states this clearly.
\begin{corollary}\label{cor:preservation}
Given a $b$-bounded and uniform \aqis\ $\P$ over an infinite
interpretation domain $U$, and an FO-CTLK formula $\varphi$, there
exists an \aqis\ $\P'$ over a finite interpretation domain $U'$ such
that $\P \models \varphi \text{ iff } \P' \models \varphi.$
\end{corollary}
It should also be noted that $U'$ can simply be taken to be any finite
subset of $U$ satisfying the cardinality requirement above. By doing
so, the finite $\oplus$-abstraction $\P'$ can be defined simply as the
restriction of $\P$ to $U'$.  Thus, every infinite, $b$-bounded and
uniform \aqis\ is bisimilar to a finite subsystem which satisfies the
same formulas. 

Note that, similarly to what noted at page
\pageref{page:nonconstructive} we are not concerned in the actual
construction of the finite abstraction. This is because we intend to
construct it directly from an artifact-centric program, as we will do
in Section~\ref{sec:as}.  Before that we explore the complexity of the
model checking problem.
 \label{abstraction}

\section{The Complexity of Model Checking Finite \aqis\ against FO-CTLK Specifications}

We now analyse the complexity of the model checking problem for finite
AC-MAS with respect to FO-CTLK specifications. The
input of the problem consists of an AC-MAS $\P$ on a finite domain $U$ and an
FO-CTLK formula $\varphi$; the solution is an assignment $\sigma$ such
that $(\P, s_0, \sigma) \models \varphi$.
Hereafter we follow \cite{Grohe01} for basic notions and definitions.
To encode an AC-MAS $\P$ we use
a tuple $E_\P=\tup{U,\D,s_0, \Phi_\tau}$,
where $U$ is the (finite) interpretation domain, 
$\D$ is the global database schema,
$s_0$ is the initial state, 
and $\Phi_\tau=\set{\phi_{\alpha_1},\ldots,\phi_{\alpha_m}}$ 
is a set of FO-formulas, each capturing the transitions 
associated with a ground action $\alpha_i$.
Since $U$ is finite, 
so is the set of ground actions, thus $\Phi_\tau$.
Each  $\varphi_{\alpha_i}$ is a FO-formula
over local predicate symbols, in both normal and ``primed'' form, 
that is, $\phi_\alpha$ can mention both $P$ and $P'$. 
For the semantics of $\Phi_\tau$, we have that
$s'\in\tau(s,\alpha)$ iff 
$s\oplus s'\models \phi_\alpha$, for $s,s'\in \D(U)$.
It can be proved that every transition relation $\tau$ can be represented in this way, 
and that, given $E_\P$, the size $||\P||\doteq\card{\S}+\card{\tau}$ of the 
corresponding \aqis\ $\P$ is at most doubly exponential in 
$||E_\P||\doteq\card{U}+||\D||+\card{\Phi_\tau}$, 
where $||\D||=\sum_{P_k\in \D} q_k$, for $q_k$ the 
arity of $P_k$. In particular, 
$||\P||=\card{\S}+\card{\tau}\leq 
2^{3\cdot 2^{||E_\P||^4}}$.


We consider the {\em combined complexity}
of the input, that is, $||E_\P|| + ||\varphi||$. In particular, we say
that the combined complexity of model checking finite AC-MAS against
FO-CTLK specifications is EXPSPACE-complete if the problem is in
EXPSPACE, i.e., there is a polynomial $p(x)$ and an algorithm solving
the problem in space bound by $2^{p(||E_\P|| + ||\varphi||)}$. We say it
is EXPSPACE-hard if every EXPSPACE problem can be reduced to model
checking finite AC-MAS against FO-CTLK specifications. We now state
the following complexity result.

\begin{theorem}\label{th:complexity}
The complexity of the model checking problem for finite AC-MAS
against FO-CTLK specifications is EXPSPACE-complete.
\end{theorem}

\textbf{Proof.} 
%
To show that the problem is in EXPSPACE, 
recall that $||\P||$ is at most doubly exponential w.r.t.~the size of the input, 
thus so is $|\S|$.  
We describe an algorithm that works in
NEXPSPACE, which combines the algorithm for model checking the
first-order fragment of FO-CTLK and the temporal epistemic fragment.
Since NEXPSPACE = EXPSPACE, the result follows.
Given an AC-MAS $\P$ and an FO-CTLK formula $\varphi$, we guess an 
assignment $\sigma$.
Given such $\sigma$, we check whether $(\P,s_0,\sigma) \models
\varphi$. This can be done by induction according to the structure of
$\varphi$. If $\varphi$ is atomic, this check can be done in
polynomial time w.r.t.~the size of the state it is evaluated on, that
is exponential time w.r.t.~$||E_\P||$.  If $\varphi$ is of the form
$\forall x \psi$, then we can apply the algorithm for model checking
first-order (non-modal) logic, which works in PSPACE. Finally, if the
outmost operator in $\varphi$ is either a temporal or epistemic
modality, then we can extend the automata-based algorithm to model
check propositional CTL in \cite{KVW00}, which works in logarithmic
space in $|\S|$. However, we remarked above that $|\S|$ is generally
doubly exponential in $||E_\P||$. Thus, if the main operator in
$\varphi$ is either a temporal or epistemic modality, then this step
can be performed in space singly exponential in $||E_\P||$. All these
steps can be performed in time polynomial in the size of $\varphi$. As
a result, the total combined complexity of model checking
finite \aqis\ is in NEXPSPACE = EXPSPACE.

To prove that the problem is EXPSPACE-hard we show a
reduction from any problem in EXPSPACE.  We assume standard definitions
of Turing machines and reductions \cite{Papadimitriou94}.
If $A$ is a problem in EXPSPACE, then there exists a deterministic
Turing machine $\T_A=\tup{\Q,\Sigma,q_0,\F,\delta}$, where $\Q$ is the
finite set of states, $\Sigma$ the machine alphabet, $q_0\in \Q$ the
initial state, $\F$ the set of accepting states, and $\delta$ the
transition function, that solves $A$ using at most space
$2^{p(|in|)}$ on a given input $in$, for some polynomial function $p$.
As standard, we assume $\delta$ to be a relation on
$(\Q\times \Sigma\times \Q\times \Sigma\times D)$, with $D=\set{L,R}$,
and $\tup{q,c,q',c',d}\in \delta$ representing a transition from state
$q$ to state $q'$, with characters $c$ and $c'$ read and written
respectively , and head direction $d$ (($L$)eft and ($R$)ight).
Without loss of generality, we assume that $\T_A$ uses only the
righthand half of the tape.
%
\newcommand{\blank}{\Box}
\newcommand{\psucc}[1]{\textsc{succ}(#1)}
\newcommand{\bin}[1]{\textsc{bin}(#1)}

From $\T_A$ and $in$, we build an  encoding 
$E_\P=\tup{\D,U,s_0,\Phi_\tau}$ 
of an \aqis\ $\P$ 
induced by a single (environment) agent 
$A_E = \tup{\D_E, L_E, Act_E, Pr_E}$ 
defined on $U=\Sigma\cup \Q\cup \set{0,1}$, 
where: 
\myi $\D_E = 
\set{P/p(\card{in})+1,Q/1,H/p(\card{in}),F/1}$; 
\myii $L_E = \D_E(U)$; 
\myiii $Act_E$
is the singleton $\{ \alpha_E \}$, with
$\alpha_E$ parameter-free; 
\myiv $\alpha_E\in Pr_E(l_E)$ for every 
$l_E\in\D(U)$.
Intuitively, the states of $\P$ correspond to configurations of
$\T_A$, while $\tau$ mimics $\delta$.
To define $E_\P$, we let $\D=\D_E$.
The intended meaning of the predicates in $\D$ 
is as follows:
the first $p(|in|)$ elements of a $P$-tuple encode 
(in binaries) the  position of a non-blank cell,
and the $(p(|in|)+1)$-th element contains 
the symbol appearing in that cell;
$Q$ contains the current state $q$ of $\T_A$;
$H$ contains the position of the cell the head 
is currently on;
$F$ contains the final states of $\T_A$, i.e., $F=\F$.
The initial state $s_0$ represents the initial configuration 
of $\T_A$, that is, for $in=in_0\cdots in_\ell$: 
$s(Q)=\set{q_0}$;
$s(H)=\set{\tup{0,\ldots,0}}$;
and $s(P)=\set{\tup{\bin{i},in_i} \mid i\in\set{0,\ldots,\ell}}$,
where $\bin{i}$ stands for the binary encoding 
in $p(\card{in})$ bits of the integer $i$.
Observe that $p(|in|)$ bits are enough to index 
the (at most) $2^{p(\card{in})}$ cells used by
$\T_A$.

As to the transition relation, we define 
$\Phi_\tau=\set{\phi_{\alpha_E}}$, where:\\
\begin{eqnarray*}
\phi_{\alpha_E} \! \! \!  &  \! \! \! =  \! \! \! &  \! \! \!
\bigvee_{\tup{q,c,q',c',d}\in\delta}(\forall x F(x)\lra F'(x))\land\\
& & Q(q)\land(\forall x Q(x)\ra x=q) \land Q'(q')\land
(\forall x Q'(x)\ra x=q')\land\\
& &  \exists\vec p (H(\vec p) \land(\forall x H(x)\ra x=\vec{p})
\land (P(\vec p, c) \vee 
(c = \blank \land \lnot\exists x P(\vec p, x)))) \land\\
& &  \exists \vec{p'} 
(d = R\ra \psucc{\vec p,\vec{p'}})\land 
(d = L\ra \psucc{\vec{p'},\vec p)}\land
H'(\vec{p'})\land (\forall x H'(x)\ra x = \vec{p'})\land\\
& &  (P'(\vec p, c')\lra(c'\neq\blank))\land (\forall x P'(\vec p, x)\ra x=c')\land\\ 
& &  (\forall \vec x,y (P(\vec x,y) \land (\vec x \neq \vec p)\ra P'(\vec x,y))\land
(\forall \vec x,y P'(\vec x,y)\ra(P(\vec x,y)\vee (\vec x = \vec p \land y = c'))))
\end{eqnarray*}

The symbol $\blank$ represents the content of blank cells, while 
$\psucc{\vec x,\vec{x'}}
=\bigwedge_{i=1}^{p(\card{in})} 
(x'_i=0\vee x'_i=1) \land 
(x'_i=1\lra ((x'_i=0\land\bigwedge_{j=1}^{i-1} x_j =1) 
\vee (x'_i=1\land\lnot\bigwedge_{j=1}^{i-1} x_j =1)))
$ is a formula capturing 
that $\vec{x'}$ is the successor of $\vec x$, for $\vec x$ and $\vec{x'}$ 
interpreted as $p(\card{in})$-bit binary encodings of integers 
(observe that  $\set{0,1}\in U$). 
Such a formula can obviously be written in  
polynomial time w.r.t.~$p(\card{in})$, as well as
$E_\P$,  and in particular $s_0$ and $\phi_{\alpha_E}$.

As it can be seen by analyzing $\Phi_\tau$, 
the obtained transition function is such that
$\tau(s,\alpha_E) = s'$ iff, 
for $\delta(q,c) = (q',c',d)$ in $\T_A$, 
we have that:
$s'(P)$ is obtained from $s(P)$ by
overwriting with $c'$ (if not blank) 
the symbol in position $(p(|in|)+1)$ of the
tuple in $s(P)$ beginning with the $p(|in|)$-tuple $s(H)$ (that is, $c$
by definition of $\phi_{\alpha_E}$); by updating $s(H)$ according to $d$,
that is by increasing or decreasing the value it contains;
and by setting $s'(Q) = \set{q'}$.  The predicate $F$ does not change.
Observe that cells not occurring in $P$ are interpreted as if
containing $\blank$ and that when $\blank$ is to be written on a cell,
the cell is simply removed from $P$.

It can be checked that,
starting with $s=s_0$,
by iteratively generating the
successor state $s'$ according to $\Phi_\tau$,  
i.e., $s'$ s.t.~$s\oplus s'\models \phi_{\alpha_E}$,
one obtains a (single) $\P$-run that is a representation of  
the computation of $\T_A$ on $in$, where each 
pair of consecutive $\P$-states corresponds to a computation step.
In particular, at each state, $Q$ contains the current state of $\T_A$. 
It should be clear that
$\varphi= EF (\exists x Q(x)\land F(x))$ holds in $\P$
iff $\T_A$ accepts $in$. 
Thus, by checking $\varphi$, we can check whether $\T_A$ accepts $in$.
This completes the proof of EXPSPACE-hardness.\qed

Note that the result above is given in terms of the ``data
structures'' in the model, i.e., $U$ and $\D$, and not the state space $\S$
itself. This accounts for the high complexity of model checking \aqis,
as the state space is doubly exponential in the size of data.

While EXPSPACE-hardness indicates intractability, we note that this is
to be expected given that we are dealing with quantified structures
which are in principle prone to undecidability. Recall also from
Section~\ref{abstraction} that the size of the interpretation domain
$U'$ of the abstraction $\P'$ is linear in the bound $b$, the number
of constants in $C$, the size of $\phi$, and $N_{Ag}$. Hence, model
checking bounded and uniform \aqis\ is EXPSPACE-complete with respect
to these elements, whose size will generally be small.  Thus, we
believe than in several cases of practical interest model
checking \aqis\ may be entirely feasible.

\medskip
We now conclude the section with some observations on 
the verification of bounded and unbounded systems.
Observe that the results presented in Sections~\ref{sec:finabs} and
\ref{sec:finabs2} apply to infinite but bounded \aqis,  i.e., whose
global states never exceed a certain size in any run. 
It is however worth noting that existential fragments of the
specification languages considered so far need not be examined with
respect to the whole \aqis. Indeed in bounded model checking
for CTLK submodels are iteratively explored until a witness for an
existential specification is found~\cite{PenczekL03}. If that
happens, we can deduce that the existential specification holds on the
full model as well. As we show below, we can extend these result to
the case of infinite \aqis.

To begin, define the $b$-restriction $\P_b$ of an \aqis\ $\P$ as
follows.
\begin{definition}[$b$-Restriction] 
Given an \aqis\ $\P = \tup{\S, U, s_0,\tau}$ and $b \in \mathbb{N}$
such that $b \geq |\adom(s_0) \cup C|$, the $b$-restriction
$\P_b = \tup{\S_b, U, s_0, \tau_b}$ of $\P$ is such that
\begin{itemize}
\item $\S_b = \{s \in \S \mid |\adom(s)| \leq b \}$;
\item  $s' \in \tau_b (s, \alpha(\vec u))$ iff $s' \in \tau (s, \alpha(\vec u))$ and $s, s' \in \S_b$.
       \end{itemize}
\end{definition}
 
Notice that $s_0 \in \S_b$ by construction and $\tau_b$ is the
restriction of $\tau$ to $\S_b$; the interpretation domain $U$ is the
same in $\P$ and $\P_b$. The result below demonstrates that if a
FO-ECTLK formula holds on the $b$-restriction, then the formula holds
on the whole \aqis.


\begin{theorem}\label{th:ectl}
Consider an \aqis\ $\P$ and its $b$-restriction $\P_b$, for
$b \in \mathbb{N}$. For any formula $\phi$ in FO-ECTLK, we have that:
\begin{eqnarray*}
\P_b \models \phi & \Rightarrow & \P \models \phi
\end{eqnarray*}
\end{theorem}

\begin{proof}
 By induction on the construction of $\phi$.  The base case for
atomic formulas and the inductive cases for propositional
connectives are trivial, as the interpretation of relation symbols
for states in $\S_b$ is the same as in $\S$.
As to the existential operators $EX$ and $EU$, it suffices to remark
that if $r$ is a run in $\P_b$ satisfying either
$EX \psi$ or $E \psi U \psi'$, then $r$ belongs to $\P$ as well by
definition of $\P_b$. 
The cases for the epistemic modalities $\bar{K}_i$ and $\bar{C}$ are similar: if
$(\P_b, s, \sigma) \models \bar{K}_i \phi$, 
then there exists $s' \in \S_b$ such that $s \sim_i s'$ and $(\P_b,
s', \sigma) \models \phi$. In particular, $s' \in \S$ and therefore
$(\P, s, \sigma) \models \bar{K}_i \phi$. For $\bar{C} \phi$ the proof
is similar by considering the transitive closure of the epistemic
relations.
Finally, the case of quantifiers follows from the fact that the active
domain for each state is the same in $\P$ and $\P_b$.
\end{proof}

Observe that there are specifications in FO-CTLK that are not
preserved from $\P_b$ to $\P$.
For instance, consider the specification $\varphi_{b} = A G \forall
x_1 , \ldots , x_{b+1} \bigvee_{i \neq j} (x_i = x_j)$ in SA-FO-CTL,
which expresses the fact that every state in every run contains at most
$b$ distinct elements.  The formula $\varphi_{b}$ is clearly satisfied by
$\P_b$ but not in $\P$, whenever $\P$ is unbounded.

Theorem~\ref{th:ectl} can in principle form the basis for an
incremental iterative procedure for checking an existential
specification $\phi$ on an infinite \aqis\ $\P$. We can begin by taking
a reasonable bound $b$ and check $\P_b \models \phi$. If that holds we
can deduce $\P \models \phi$; if not we can increase the bound and
repeat. The procedure is sound but clearly not complete.  As mentioned
earlier, this is in spirit of bounded model
checking~\cite{biereCCSZ03}. Here, however, the bound is on the size
of the states, rather than the length of the runs.

\section{Model Checking Artifact-Centric Programs}\label{sec:results}\label{sec:as}
We have so far developed a formalism that can be used to specify and
reason about temporal-epistemic properties of models representing
artifact-centric systems. We have identified two notable classes that
admit finite abstractions. As we remarked in the introduction,
however, artifact-centric systems are typically implemented through
declarative languages such as GSM~\cite{Hulletal11}. It is therefore of
paramount interest to investigate the verification problem, not just
on a Kripke semantics such as \aqis, but on concrete programs. As
discussed, while GSM is a mainstream declarative language for
artifact-centric environments, alternative declarative approaches
exist. In what follows for the sake of generality we ground our
discussion on a very wide class of declarative languages and define
the notion of \emph{artifact-centric program}. Intuitively, an
artifact-centric program (or AC program) is a declarative description
of a whole multi-agent system, i.e., a set of services, 
that interact with the artifact system
(see discussion in the Introduction). Since artifact systems are also
typically implemented declaratively (see~\cite{HHV11}) in what
follows AC programs will be used to encode both the artifact system
itself and the agents in the system. This also enables us to import
into the formalism the previously discussed features of views and
windows typical in GSM and other languages.

This section is organised as follows. Firstly, we define AC programs
and give their semantics in terms of \aqis. Secondly, we show that
any \aqis\ that results from an AC program is uniform. This enables us
to state that, as long as the generated \aqis\ is bounded, any AC
program admits an \aqis\ as its finite model. In this context it
is actually important to give constructive procedures for the
generation of the finite abstraction; we provide such a procedure
here. This enables us to state that, under the assumptions we
identify, AC programs admit decidable verification by means of model
checking their finite model.



We start by defining the abstract syntax of AC programs.
\begin{definition}[AC Program]\label{def:acp}
An {\em artifact-centric program} (or AC program) is a tuple $\ACP
	= \tup{\D,U,\Sigma}$, where:
\begin{itemize}
\item $\D$ is the program's database schema;
\item $U$ is the program's interpretation domain;
\item $\Sigma=\set{\Sigma_0,\ldots,\Sigma_n}$ is the set of 
{\em agent programs} $\Sigma_i=\tup{\D_i, l_{i0}, \Omega_i}$,
where: 
\begin{itemize} \item $\D_i\subseteq \D$ is {\em agent $i$'s
database schema}, s.t.~$\D_i\cap\D_j=\emptyset$, for $i\neq j$; 
\item $l_{i0}\in \D_i(U)$ is {\em agent $i$'s initial state (as a database
instance)}; 
\item $\Omega_i$ is the set of {\em local action
descriptions} in terms of preconditions and postconditions of the form
$\alpha(\vec x)\doteq\tup{\pi(\vec y),\psi(\vec z)}$, where:
\begin{itemize}
\item $\alpha(\vec x)$ is the {\em action signature} and $\vec x
= \vec y\cup \vec z$ is the set of its {\em parameters};
\item $\pi(\vec y)$ is the {\em action precondition}, i.e., an FO-formula over
$\D_i$; 
\item $\psi(\vec z)$ is the {\em action postcondition}, i.e., an
FO-formula over
$\D\cup\D'$.  \end{itemize} \end{itemize} \end{itemize}
\end{definition}

Recall that local database schemas and instances were introduced in
Definition~\ref{def:ags}.
Observe that AC programs are defined modularly by giving the agents'
programs including preconditions and postconditions as well as those
of the environment.

Notice that preconditions use relation symbols from the local
database only, while postconditions can use any symbol from the whole
$\D$. This accounts for the intuition formalised in \aqis\ as well as
present in temporal-epistemic logic literature that agents' actions
may change the environment and the state of other agents.
For an action $\alpha(\vec x)$, we let
$\const(\alpha)=\const(\pi)\cup\const(\psi)$, 
$\vars(\alpha)=\vars(\pi)\cup \vars(\psi)$, and 
$\free(\alpha)=\vec x$.
An {\em execution} of $\alpha(\vec x)$ with {\em ground parameters}
$\vec u\in U^{|\vec x|}$ is the {\em ground} action $\alpha(\vec
u)=\tup{\pi(\vec v),\psi(\vec w)}$, where $\vec v$ (resp.~$\vec w$) is
obtained by replacing each $y_i$ (resp.~$z_i$) with the value
occurring in $\vec u$ at the same position as $y_i$ (resp.~$z_i$) in
$\vec x$.
Such replacements make both 
$\pi(\vec v)$ and $\psi(\vec w)$ ground.
Finally, we define
the set $C_{\ACP}$ of all constants mentioned in 
$\ACP$, i.e.,
$C_{\ACP} =\bigcup_{i=1}^n\big(\adom(D_{i0}) \cup
\bigcup_{\alpha \in\Omega_i} \const(\alpha)\big)$.

The semantics of a program is given in terms of the \aqis\ induced by
the agents that the program implicitly defines.  Formally, this is
captured by the following definition.
\begin{definition}[Induced Agents]\label{def:ind-ag}
Given an AC program $\ACP = \tup{\D,U,\Sigma}$, an
agent \emph{induced} by $\ACP$ is a tuple
$A_i=\tup{\D_i,L_i,Act_i,Pr_i}$ on the interpretation domain $U$ such
that, for $\Sigma_i=\tup{\D_i,l_{i0},\Omega_i}$:
\begin{itemize} 
\item $L_i\subseteq \D_i(U)$ is the set of the agent's local states; 
\item $Act_i=\set{\alpha(\vec x)\mid \alpha(\vec x)\in \Omega_i}$ is
the set of local actions; 
\item The protocol $Pr_i(l_i)$ is defined by $\alpha(\vec u)\in
Pr_i(l_i)$ iff $l_i\models \pi(\vec v)$ for $\alpha(\vec u)=\tup{\pi(\vec
	v),\psi(\vec w)}$.  
\end{itemize}
\end{definition}
Note that the definition of induced agent is in line with the
definition of Agents (Definition~\ref{def:ags}). 
Agents induced as above are composed to give an \aqis\ associated with
an AC program.
\begin{definition}[Induced \aqis]\label{def:ind-aqis}
Given an AC program $\ACP$ and the set $Ag=\set{A_0,\ldots,A_n}$ of
agents induced by $\ACP$, the \aqis\ \emph{induced} by $\ACP$ is
the tuple $\P_{\ACP} =\tup{\S, U, s_0,\tau}$, where: 
\begin{itemize}
\item $\S \subseteq L_0\times\cdots\times L_n$ is the set
of \emph{reachable states};
\item $s_0 = \tup{l_{00},\ldots,l_{n0}}$ is the {\em initial global
state}; 
\item $U$ is the interpretation domain;
\item $\tau$ is the global transition function defined by the
following condition: $s' \in \tau(s,\tup{\alpha_1(\vec
u_1),\ldots,$ $\alpha_n(\vec u_n)})$, with $s=\tup{l_0,\ldots,l_n}$ and
$\alpha_i(\vec u_i)=\tup{\pi_i(\vec v_i),\psi_i(\vec w_i)}$
($i\in\set{0,\ldots,n}$), iff the following conditions are satisfied: 
\begin{itemize} 
\item for every $i\in\set{0,\ldots,n}$, $l_i\models \pi_i(\vec v_i)$; 
\item $\adom(s')\subseteq \adom(s)\cup \bigcup_{i=0,\ldots,n}\vec
w_i\cup\const (\psi_i)$; \item $D_s\oplus D_{s'}\models \psi_i(\vec
w_i)$,
 where $D_s$ and $D_{s'}$ are obtained from $s$ and $s'$ as
discussed on p.~\pageref{refref}.  
\end{itemize} 
\end{itemize}
\end{definition}
Given an AC program, the induced \aqis\ is the Kripke model
representing the whole execution tree for the AC program and
representing all the data in the system.
Observe that all actions performed are enabled by the respective
protocols and that transitions can introduce only a bounded number of
new elements in the active domain, those bound to the action
parameters.
It follows from the above that AC programs are \emph{parametric} with
respect to the interpretation domain, i.e., by replacing the
interpretation domain we obtain a different \aqis.
For simplicity, we assume that for every postcondition $\psi$ in a
program, if a predicate does not occur in the postcondition, it is
left unchanged by the relevant transitions.  Formally, this means that
we implicitly add a conjunct of the form $\forall \vec{x} P(\vec
x) \lra P'(\vec x)$ $(*)$ to the postcondition whenever $P$ is not
mentioned in $\psi$.
Further, we assume that every program induces an \aqis\ whose
transition relation is serial, i.e., \aqis\ states always have
successors. These are basic requirements that can be easily fulfilled,
for instance, by assuming that each agent has a \textsf{skip} action
with an empty precondition and a postcondition of the form $(*)$ for
every $P \in \D$. In the next section we present an example of one
such program.

A significant feature of AC programs is that they induce
uniform \aqis.
\begin{lemma}\label{lem:spec-unif}
	Every \aqis\ $\P$ induced by an AC program $\ACP$ is uniform.
\end{lemma}

\begin{proof}
By Prop.~\ref{prop1}, it is sufficient to consider only the temporal
transition relation $\to$, as $\adom(s_0)\subseteq C_{\ACP}$.
Consider $s,s',s''\in\S$ and $s''' \in L_0\times\cdots\times L_n$,
such that $s\oplus s'\iso s''\oplus s'''$ for some witness
$\iota$. Also, assume that
there exists $\vec\alpha(\vec u)=\tup{\alpha_1(\vec
u_1),\ldots,\alpha_n(\vec u_n)}\in Act(U)$ such
that~$s'\in\tau(s,\vec\alpha(\vec u))$.  We need to prove that for
every constant-preserving bijection $\iota'$ that extends $\iota$ to
$\vec u$, we have that $s''' \in\tau(s'',\vec\alpha(\iota'(\vec u)))$.
To this end, we remark that any witness $\iota$ for
$s\oplus s'\iso s''\oplus s'''$ can be extended to an injective
function $\iota'$ on $\bigcup_{i \in Ag} \vec u_i$.  Obviously, $U$
contains enough distinct elements for $\iota'$ to exist, as every
$\vec u_i$ takes values from $U$.  Now, by an argument analogous to
that of Proposition~\ref{prop:iso-inst-eq}, it can be seen that for
any $FO$-formula $\varphi$ and equivalent assignments $\sigma$ and
$\sigma'$, we have that $(s\oplus s',\sigma)\models \varphi$ iff
$(s''\oplus s''',\sigma')\models \varphi$.  But then, this holds, in
particular, for $\sigma'$ obtained from $\sigma$ by applying $\iota'$
to the values assigned to each parameter, i.e., $\iota'(\vec u)$, and
for the pre- and postconditions of all actions involved in the
transition $s \xrightarrow{\alpha(\vec u)} s'$.
Thus, we have
$s'''\in\tau(s'',\vec\alpha(\iota'(\vec u)))$, i.e., 
$\P$ is uniform.
\end{proof}

We can now define what it means for an AC program to satisfy a
specification, by referring to its induced \aqis.

\begin{definition}\label{def:acp-satisfaction} 
Given an AC program $\ACP$, a FO-CTLK formula
$\varphi$, and an assignment $\sigma$, we say that {\em $\ACP$
satisfies $\varphi$ under $\sigma$}, written
$(\ACP,\sigma)\models\varphi$, iff
$(\P_{\ACP},s_0,\sigma)\models \varphi$.
\end{definition}
It follows that the model checking problem for an AC program
against a specification $\phi$ is defined in terms of the model
checking problem for the \aqis\ $\P_{\ACP}$ against $\phi$.

The following result allows us to reduce the verification of
any AC program with an infinite interpretation domain $U_1$, that
induces a $b$-bounded \aqis, to the verification of an AC program
over a finite $U_2$.  To show how it can be done, we let
$N_{\ACP}=\sum_{i \in \set{1,\ldots,n}} \max_{\alpha(\vec
x)\in\Omega_i}\set{\card{\vec x}}$ be the maximum number of different
parameters that can occur in a joint action of $\ACP$.
\begin{lemma}\label{lem:finite-abstraction}
Consider an AC program $\ACP_1 = \tup{\D,U_1,\Sigma}$ operating on an
infinite interpretation domain $U_1$ and assume its induced \aqis\
$\P_{\ACP_1}=\tup{\S_1, U_1, s_{10},\tau_1}$ is $b$-bounded. Consider a
finite interpretation domain $U_2$ such that $C_{\ACP_1}\subseteq U_2$
and $\card{U_2}\geq 2b+\card{C_{\ACP_1}}+N_{\ACP_1}$ and the AC program
$\ACP_2 = \tup{\D,U_2,\Sigma}$. Then, the \aqis\ $\P_{\ACP_2}=\tup{\S_2,
U_2, s_{20},\tau_2}$ induced by $\ACP_2$ is a finite abstraction of
$\P_{\ACP_1}$.
\end{lemma}
\begin{proof}
Let $Ag_1$ and $Ag_2$ be the set of agents induced respectively 
by $\ACP_1$ and $\ACP_2$,  according to Def.~\ref{def:ind-ag}. 
First, we prove that the set of agents
$Ag_1$ and $Ag_2$ satisfy Def.~\ref{def:ab-agents},
for $Ag=Ag_1$ and $Ag'=Ag_2$.
To this end, observe that because $\ACP_1$ and $\ACP_2$ differ only in
$U$, by Def.~\ref{def:ind-ag}, $\D=\D'$, $L_i'\subseteq \D_i'(U')$,
and $Act'=Act$.  Thus, only requirement~\ref{req:ab-agents-prtcl} of
Def.~\ref{def:ab-agents} still needs to be proved.
To see it, fix $i\in\set{1,\ldots,n}$ and assume that $\alpha(\vec u)\in Pr_i(\l_i)$.
By Def.~\ref{def:ind-ag}, we have that $l_i\models \pi(\vec v)$, for 
$\alpha(\vec u)=\tup{\pi(\vec v),\psi(\vec w)}$.
By the assumption on $\card{U_2}$, since $\const(\alpha)\subseteq C_{\ACP_1}\subseteq U_2$,
$\card{\vec u}\leq N_{\ACP_1}$, and $\card{\adom(l_i)}\leq b$, 
we can define an injective function $\iota:\adom(l_i)\cup \vec u\cup C_{\ACP_1}\mapsto U_2$
that is the identity on $C_{\ACP_1}$. 
Thus, for $l_i'=\iota(l_i)$, we can easily extract from $\iota$ a witness for $l_i\simeq l'_i$.
Moreover, it can be seen that $\vec v$ and $\vec v'$ are equivalent for 
$\pi$.
Then, by applying Prop.~\ref{prop:iso-inst-eq} 
to $l_i$ and $l_i'$, we conclude that $l'_i\models \pi(\vec v')$,
for $\vec v'=\iota(\vec v)$. Hence, by Def.~\ref{def:ind-ag},
$\alpha(\vec u')\in Pr'_i(l'_i)$. 
So, we have shown the right-to-left part of 
requirement \ref{req:ab-agents-prtcl}.
The left-to-right part can be shown similarly and in a simplified way
as $U_1$ is infinite.

Thus, we have proven that $Ag =Ag_1$ and $Ag'=Ag_2$ are obtained as in
Def.~\ref{def:ab-agents}. Hence, the assumption on $Ag$ and $Ag'$ in
Def.~\ref{ref:abs} is fulfilled.  We prove next that also the
remaining requirements of Def.~\ref{ref:abs} are satisfied.
Obviously, since $\Sigma$ is the same for $\ACP_1$ and $\ACP_2$, by
Def.~\ref{def:ind-aqis}, $s_{10}=s_{20}$, so the initial states of
$\P_{\ACP_1}$ and $\P_{\ACP_2}$ are the same.  It remains to show that
the requirements on $\tau_1$ and $\tau_2$ are satisfied. We prove the
right-to-left part.  To this end, take two states
$s_1=\tup{l_{10},\ldots,l_{1n}}$,
$s_1'=\tup{l'_{10},\ldots,l'_{1n}}$ in $\S_1$ and a joint action
$\vec \alpha(\vec u)=\tup{\alpha_0(\vec u_0),\ldots,\alpha_n(\vec
u_n)}\in Act(U)$ such that $s_1'\in\tau_1(s_1,\vec\alpha(\vec u))$.
Consider $s_1\oplus s_1'$.  By the assumptions on $U_2$, there exists
an injective function $\iota:\adom(s_1)\cup\adom(s'_1)\cup \vec u\cup
C_{\ACP_1}\mapsto U_2$ that is the identity on $C_{\ACP_1}$ (recall
that $\card{\adom(s_1)},\card{\adom(s'_1)}\leq b$).  Then, for
$s_2=\tup{\iota(l_{10}),\ldots, \iota(l_{1n})}$,
 $s_2'=\tup{\iota(l'_{10}),\ldots, \iota(l'_{1n})}$ in $\S_2$, we can
extract, from $\iota$, a witness for $s_1\oplus s'_1\simeq s_2\oplus
s_2'$.  Moreover, it can be seen that for every $\pi_i$ and $\psi_i$
in $\vec \alpha_i(\vec x_i)=\tup{\pi_i(\vec y_i),\psi_i(\vec z_i)}$,
$\vec u$ and $\vec u'=\iota(\vec u)$ are equivalent with respect to
$s_1\oplus s'_1$ and $s_2\oplus s_2'$.
Now, consider Def.~\ref{def:ind-aqis} and recall that both $\P_{\ACP_1}$
and $\P_{\ACP_2}$ are \aqis\ induced by $\ACP_1, \ACP_2$, respectively.  By
applying Prop.~\ref{prop:iso-inst-eq}, we have that, for
$i\in\set{0,\ldots,n}$: $\iota(l_{1i})\models \pi_i(\iota(\vec v_i))$
iff $l_{1i}\models \pi_i(\vec v_i)$; $D_{s_2}\oplus
D_{s_2'}\models \psi_i(\iota(\vec w_i))$ iff $D_{s_1}\oplus
D_{s_1'}\models \psi_i(\vec w_i)$.
In addition, by the definition of $\iota$,
$\adom(s_1')\subseteq  \adom(s_1)\cup \bigcup_{i=0,\ldots,n}\vec w_i\cup\const (\psi_i)$
iff 
$\adom(s_2')\subseteq  \adom(s_2)\cup \bigcup_{i=0,\ldots,n}\iota(\vec w_i)\cup\const (\psi_i)$.
But then, it is the case that $s'_2\in \tau_2(s_2',\vec \alpha(\iota(\vec u_0),\ldots,\iota(\vec u_n)))$.
So we have proved the right-to-left part of the second requirement of
Def.~\ref{ref:abs}. The other direction follows similarly.  Therefore,
$\P_{\ACP_2}$ is an abstraction of $\P_{\ACP_1}$.
\end{proof}

Intuitively, Lemma~\ref{lem:finite-abstraction} shows that the
following diagram commutes, where $[U_1/U_2]$ stands for the
replacement of $U_1$ by $U_2$ in the definition of $\ACP_1$. Observe
that since $U_2$ is finite, one can actually apply
Def.~\ref{def:ind-aqis} to obtain $\P_{\ACP_2}$, while this cannot be
done for $\ACP_1$, as $U_1$ is infinite.
	\begin{displaymath}
	\xymatrix{	
	\ACP_1~ 
		\ar[rr]^{\text{ Def.~\ref{def:ind-aqis}}}  
		\ar[d]^{[U_1/U_2]}
		& & \P_{\ACP_1} \ar[d]^{\text{Def.~\ref{ref:abs}}} \\
	\ACP_2~\ar[rr]_{\text{ Def.~\ref{def:ind-aqis}}} 	& & \P_{\ACP_2} }
	\end{displaymath}

The following result, a direct consequence of Lemma~\ref{theor1} and
Lemma~\ref{lem:finite-abstraction}, is the key conclusion of this
section.

\begin{theorem}\label{th:spec-abstr}
Consider an FO-CTLK formula
$\varphi$, an AC program $\ACP_1$ operating on an infinite interpretation
domain $U_1$ and assume its induced \aqis\ $\P_{\ACP_1}$ is $b$-bounded.
Consider a finite interpretation domain $U_2$ such that
$C_{\ACP_1}\subseteq U_2$ and $\card{U_2}\geq
2b+\card{C_{\ACP}}+\max\set{N_{\ACP},\card{\vars(\varphi)}}$, and the AC
program $\ACP_2 = \tup{\D,U_2,\Sigma}$. Then
we have that:
\begin{eqnarray*} \ACP_1 \models \varphi & \text{iff} &
\ACP_2 \models \varphi.  \end{eqnarray*}
\end{theorem}
\textbf{Proof.}
By Lemma~\ref{lem:finite-abstraction} $\P_{\ACP_2}$ is a finite
abstraction of $\P_{\ACP_1}$. Moreover, $\card{U_2}\geq
2b+\card{C_{\ACP}}+\max\set{N_{\ACP},\card{\vars(\varphi)}}$ implies
$\card{U_2} \geq 2b + \card{C_{\ACP}} + \card{\vars(\varphi)}$. Hence,
we can apply Lemma~\ref{theor1} and the result follows. \qed

The above is the key result in this section. It shows that if the generated
\aqis\ model is bounded, then any AC program can be verified by model
checking its finite $\oplus$-abstraction, i.e., a
$\oplus$-bisimilar \aqis\ defined on a finite interpretation
domain. Note that in this case the procedure is entirely constructive:
given an AC program $\ACP_1 = \langle \D, U_1, \Sigma \rangle$ on an
infinite domain $U_1$ and an FO-CTLK formula $\varphi$, to check whether
$\ACP_1$ satisfies the specification $\varphi$, we first consider the
finite ``abstraction'' $\ACP_2 = \langle \D, U_2, \Sigma \rangle$
defined on a finite domain $U_2$ satisfying the requirement on
cardinality in Theorem~\ref{th:spec-abstr}. Since $U_2$ is finite,
also the induced \aqis\ $\P_{\ACP_2}$ is finite, hence we can apply
standard model checking techniques to verify whether $\P_{\ACP_2}$
satisfies $\varphi$.  Finally, by definition of satisfaction for AC
programs and Theorem~\ref{th:spec-abstr}, we can transfer the result
obtained to decide the model checking problem for the original
infinite AC program $\ACP_1$ and $\varphi$.



Also observe that in the finite abstraction considered above the
abstract interpretation domain $U_2$, depends on the number of
distinct variables that the specification $\varphi$ contains.
Thus, in principle, to check the same AS program against a different
specification $\varphi'$, one should construct a new abstraction
$\P_{\ACP'_2}$ using a different interpretation domain $U_2'$, and then
check $\varphi'$ against it.  However, it can be seen that if the
number of distinct variables of $\varphi'$ does not exceed that of
$\varphi$, the abstraction $\P_{\ACP_2}$, used to check $\varphi$, can
be re-used for $\varphi'$.
Formally, let FO-CTLK$_k$ be the set of all FO-CTLK formulas
containing at most $k$ distinct variables.  We have the following
corollary to Theorem~\ref{th:spec-abstr}.
\begin{corollary}
	If $\card{U_2}\geq 2b+\card{C_{\ACP}}+\max\set{N_{\ACP},k}$,
	then, for every FO-CTLK$_k$ formula $\varphi$,
	$\ACP_1 \models \varphi$ iff $\ACP_2 \models \varphi$.
\end{corollary}

This result holds in particular for $k=N_{\ACP}$; thus for 
FO-CTLK$_{N_{\ACP}}$ formulas, we have an abstraction procedure
that is specification-independent.

Theorem~\ref{th:spec-abstr} requires the induced \aqis\ to be bounded,
which may seem a difficult condition to check a priori. Note however
that AC programs are declarative. As such it is straightforward to
give postconditions that enforce that no transition will generate
states violating the boundedness requirement. The scenario in the next
section will exemplify this.

\section{The Order-to-Cash Scenario}

In this section we exemplify the methodology presented so far in the
context of a business process inspired by an IBM customer
use-case~\cite{Hulletal11}. 
The {\em order-to-cash} scenario describes the actions performed by a
number of agents in an e-commerce situation relating to the purchase
and delivery of a product.
The agents in the system consist of a {\em manufacturer}, some {\em
customers}, and some {\em suppliers}.  The process begins when a
customer prepares and submits a {\em purchase order} (PO), i.e., a
list of products the customer requires, to the manufacturer.
Upon receiving a PO, the manufacturer prepares a {\em material order}
(MO), i.e., a list of components needed to assemble the requested
products.  The manufacturer then selects a supplier and forwards him
the relevant material order.
%
Upon receipt a supplier can either accept or reject a MO. In the
former case he then proceeds to deliver the requested components to
the manufacturer. In the latter case he notifies the manufacturer of
his rejection.
If an MO is rejected, the manufacturer can delete it and then prepare
and submit new MOs.
When the components required have been delivered to the manufacturer,
he assembles the product and, provided the order has been paid for,
he delivers it to the customer.
Any order which is directly on indirectly related to a PO can be
deleted only after the PO is deleted.


We can encode the {\em order-to-cash} business process as an
 artifact-centric program $ACP_{otc}$, where the artifact data models
 are represented as database schemas and its evolution is
characterised by an appropriate set of operations.
It is natural to identify 2 classes of artifacts, representing the PO
and the MO, each corresponding to the respective orders by the agents.
An intuitive representation of the artifact lifecycles, i.e., the
evolution of some key records in the artifacts' states, capturing only
the dependence of actions from the artifact statuses, is shown in
Fig.~\ref{fig:lifecycles}.  Note that this is an incomplete
representation of the business process, as the interaction between
actions and the artifact data content is not represented.
\begin{figure}
	\centering
	\subfigure[Purchase Order lifecyle\label{subfig:cpo-lifecycle}]{
		\begin{tikzpicture}
			\begin{scope}  
				\scriptsize 
				\node (init) {};
                                	\node[state,right of=init,,xshift=.2cm] (prepared) {$prepared$}; 
                                 	\node[state,right of=prepared,xshift=1.2cm] (pending) {$pending$}; 
                                 	\node[state,right of=pending] (paid) {$paid$}; 
                                 	\node[state,right of=paid,xshift=.4cm] (shipped) {$shipped$}; 
	                         	\node[state,fill,right of=shipped,xshift=.1cm] (deleted) {}; 
	
				\path 
					(init) edge node[above] {$createPO$} (prepared)
					(prepared) edge node[above] {$submitPO$} (pending)
					(pending) edge node[above] {$pay$} (paid)
					(paid) edge node[above] {$shipPO$} (shipped)
					(shipped) edge node[above] {$deletePO$} (deleted);
			\end{scope}
		\end{tikzpicture}
	}
	\subfigure[Material Order lifecyle\label{subfig:mpo-lifecycle}]{
		\begin{tikzpicture}
			\begin{scope}  
				\scriptsize
				\node (init) {};
				\node[state,right of=init,xshift=.4cm] (preparation) {$preparation$};
				\node[state,right of=preparation,xshift=1.4cm] (submitted) {$submitted$};
				\node[state,right of=preparation,above of=submitted,yshift=-1.2cm] (accepted) {$accepted$};
				\node[state,right of=accepted,xshift=.6cm] (shipped) {$shipped$};
				\node[state,right of=preparation,below of=submitted,yshift=1.2cm] (rejected) {$rejected$};
				\node[state,fill,right of=rejected,xshift=.6cm] (deleted) {};

				\path
					(init) edge node[above] {$createMO$} (preparation) 
					(preparation) edge[] node[above] {$doneMO$} (submitted)
					(submitted) edge node[left] {$acceptMO$} (accepted)
					(submitted) edge node[right] {$rejectMO$} (rejected)
					(accepted) edge node[above] {$shipMO$} (shipped)
					(shipped) edge node[right] {$deleteMO$} (deleted)
					(rejected) edge node[above] {$deleteMO$} (deleted);
			\end{scope}
		\end{tikzpicture}
	}
\caption{Lifecycles of the artifacts involved in the order-to-cash scenario.\label{fig:lifecycles}}
\end{figure}
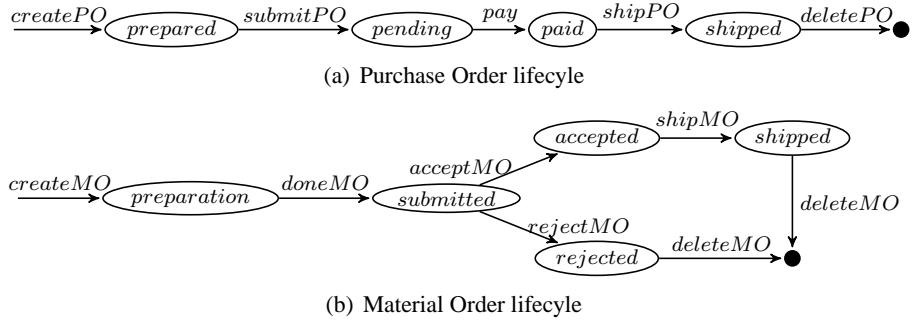

Next, we encode the whole system as an AC program, where the artifact
data models are represented as a relational database schema, and the
corresponding lifecycles are formally characterised by an appropriate
set of actions.
We reserve a distinguished relation for each artifact class.
In addition, we introduce {\em static} relations to store
product and material information.
For the sake of presentation we assume to be dealing with three agents
only: one customer $c$, one manufacturer $m$ and one supplier $s$.
The database schema $\D_i$ for each agent $i \in \{c, m , s\}$ can
therefore be given as:
\begin{itemize}
\item Customer $c$: $\D_c = \{{\it Products(prod\_code, budget)}, {\it PO(id, prod\_code, offer, status)}\};$	
\item Manufacturer $m$:  $\D_m =\{ {\it PO(id, prod\_code, offer, status)}, {\it MO(id, prod\_code, price, status)}\};$
\item Supplier $s$: $\D_s = \{{\it Materials(mat\_code, cost)}, {\it
MO(id, prod\_code, price, status)}\}.$
\end{itemize}

The relations
{\it Products} and {{\it Materials}, as well as {\it PO} and {\it MO}
are self-explanatory. Note the presence of the attribute {\it status}
in the relations corresponding to artifacts.

As interpretation domain, we consider the infinite set $U_{otc}$ of
alphanumeric strings. Also, we assume that in the initial state
the only non-empty relations are
{\it Products} and {\it Materials}, which contain background
information, such as
the catalogue of available products.

Hence, the artifact-centric program $ACP_{otc}$ corresponding to the
order-to-cash scenario can be given formally as follows:
\begin{definition}
The {\em artifact-centric program} $ACP_{otc}$ is a tuple
$\tup{\D_{otc},U_{otc},\Sigma_{otc}}$, where:
\begin{itemize}
\item  the program's database schema $\D_{otc}$ and interpretation domain $U_{otc}$ are introduced as above, i.e., $\D_{otc} = \D_c \cup \D_m \cup \D_s = \{{\it Products}/2, {\it PO}/4, {\it MO}/4, {\it Materials}/2 \}$  and $U_{otc}$ is the set of all alphanumeric strings.
\item $\Sigma= \set{\Sigma_c, \Sigma_m, \Sigma_s  }$ is the set of 
{\em agent specifications} for the customer $c$, the manufacturer $m$ and the supplier $s$. Specifically, for each $i \in \{c, m, s \}$, $\Sigma_i=\tup{\D_i, l_{i0}, \Omega_i}$ is such that: 
\begin{itemize} \item $ \D_i\subseteq \D$ is  agent $i$'s
database schema as detailed above, i.e., $\D_c = \{{\it Products}/2, {\it PO}/4 \}$, $\D_m = \{{\it PO}/4, {\it MO}/4 \}$, and $\D_s = \{{\it MO}/4, {\it Materials}/2 \}$. 
\item $l_{c0}$, $l_{m0}$, and $l_{s0}$ are database instances in $\D_c(U_{otc})$, $\D_m(U_{otc})$, and $\D_s(U_{otc})$ respectively s.t.~$l_{c0}({\it Products})$ and $l_{s0}({\it Materials})$ are not empty, i.e., they contain some background information, while $l_{c0}({\it PO})$, $l_{m0}({\it PO})$, $l_{m0}({\it MO})$ and $l_{s0}({\it MO})$ are empty. 
\item We assume that $\Omega_c$ contains the actions {\it createPO(prod\_code,offer)}, {\it submitPO(po\_id)},  {\it pay(po\_id)},  {\it deletePO(po\_id)}.
Similarly, $\Omega_m = \{ {\it createMO(po\_id,price)},$\newline $ {\it
doneMO(mo\_id)}, {\it shipPO(po\_id)}, {\it deleteMO(mo\_id)} \}$ and
$\Omega_s = \{ {\it acceptMO (mo\_id)},$ \newline $ {\it rejectMO(mo\_id)}, {\it
shipMO(mo\_id)} \}$.
\end{itemize}
\end{itemize}
\end{definition}

System actions capture {\em legal} operations on the underlying
database and, thus, on artifacts.  In Table~\ref{tab:action-specs} we
report some of their specifications.  Variables (from $V$) and
constants (from $U$) are distinguished by fonts $v$ and ${\sf c}$,
respectively.  From Section~\ref{sec:as} we adopt the convention that an
action affects only those relations whose name occurs in $\psi$.
\begin{table}
\caption{Specification of the actions affecting the artifacts PO and  MO in the order-to-cash scenario.\label{tab:action-specs}}
\newcommand{\tab}[0]{\phantom{xx}}
\begin{itemize}
\item $createPO(prod\_code) = \tup{\pi(prod\_code),\psi(prod\_code)}$, where:	
		\begin{itemize}
			\item $\pi(prod\_code)\equiv
 \exists b~Products(prod\_code, b)$
			\item $\psi(prod\_code) \equiv \exists
                          id, b~(PO'(id,prod\_code, b, {\sf prepared})
                          \land$\\ 
                          $Products(prod\_code, b) \land$\\ 
$\forall id', pc, o, s~(PO(id', pc, o,s) \ra
                          id \neq id'))$
\end{itemize}
\item $createMO(po\_id, price) = \tup{\pi(po\_id,price),\psi(po\_id,price)}$, where:	
		\begin{itemize}
			\item $\pi(po\_id,price)\equiv \exists pc,o~(PO(po\_id,pc,o,{\sf prepared})$
			\item $\psi(po\_id,price)\equiv
                          (MO'(po\_id,pc,price,{\sf preparation}) \land$\\
$ \exists o PO(po\_id,pc,o,{\sf prepared}) \land$\\
$ \forall id',pc,pr,s~(MO(id',pc,pr,s) \ra
                          id \neq id'))$
\end{itemize}
	\item $doneMO(mo\_id)=\tup{\pi(mo\_id),\psi(mo\_id)}$, where:
		\begin{itemize}
			\item $\pi(mo\_id)\equiv\exists
                          pc, p~MO(mo\_id,pc,pr, {\sf
                            preparation})$
			\item $\psi(mo\_id)\equiv 
				\forall w,pc,p,s~\big(
					(w\neq mo\_id \ra(MO(w,pc,p,s)\lra
                                MO'(w,pc,p,s))) \land$\\ 
$(MO(mo\_id,pc,p,s) \ra 						(MO'(mo\_id,pc,p,{\sf submitted})\land$\\ 
						$(s\neq{\sf submitted}\ra\lnot MO'(mo\_id,pc,p,s))))
				\big)$
		\end{itemize}		
	\item $acceptMO(mo\_id)=\tup{\pi(mo\_id),\psi(mo\_id)}$, where:
		\begin{itemize}
			\item $\pi(mo\_id)\equiv \exists
                          mo\_id, pc, p~MO(mo\_id, pc, pr, {\sf
                            submitted}) \land Material(pc,p)$
			\item $\psi(mo\_id)\equiv 
				\forall w,pc,p,s~\big(
					(w\neq mo\_id \ra(MO(w,pc,p,s)\lra
                                MO'(w,pc,p,s))) \land$\\ $(MO(mo\_id,pc,p,s) \ra 						(MO'(mo\_id,pc,p,{\sf accepted})\land$\\ 
						$(s\neq{\sf accepted}\ra\lnot MO'(mo\_id,pc,p,s))))
				\big)$
		\end{itemize}	
\end{itemize}
\end{table}

Consider, for instance, the action {\it createPO} performed by the
customer $c$, whose purpose is the creation of a PO artifact instance
related to a given $prod\_code$.  Its precondition requires that the
action parameter $prod\_code$ refers to an actual product in the
$Products$ database; while the postcondition guarantees that the {\it
offer} value in PO is set equal to {\it budget} as well as the $id$ of
the new PO is unique.
As regards the action {\it createMO}, performed by the manufacturer
$m$ and meant to create instances of MO artifacts, its precondition
requires that $po\_id$ is the identifier of some existing PO.
Its postcondition states that, upon execution, the {\it MO} relation
contains exactly one additional tuple, with identifier attribute set
to $id$, with attribute $status$ set to ${\sf preparation}$ and asking
price set to $price$.
As an example of action triggering an artifact's status transition,
consider the action {\it doneMO} performed also by the manufacturer
$m$.
{\it doneMO} is executable only if the MO artifact is in status ${\sf
preparation}$; its effect is to set the status attribute to ${\sf
submitted}$.
Finally, as an example of an action triggered by a choice, consider
the action {\it acceptMO} performed by the supplier $s$. It is
triggered only if the entries for the product code $pc$ and the price
$p$ have matching values in the {\it Materials} database. The action
outcome is to set the status attribute to ${\sf accepted}$.

Notice that although actions are typically conceived to manipulate
artifacts of a specific class their preconditions and postconditions 
may depend on artifact instances of different classes.
For example note that the action {\it
createMO} manipulates MO artifacts, but its preconditions and
postconditions may depend on artifact instances originating from
different classes (e.g. {\it createMO}'s precondition depends on PO
artifacts).
We stress that action executability depends not only on the status
attribute of an artifact, but on the data content of the whole
database, i.e., of all other artifacts. Similarly, action executions
affect not only status attributes.
Most importantly, by using first-order formulas such as $\phi_b
= \forall x_{1}, \ldots , x_{b+1} \bigvee_{i \neq j} (x_i = x_j)$ in
the postcondition $\psi$, we can guarantee that the AC program in
question is bounded and is therefore amenable to the abstraction
methodology of Section~\ref{sec:results}.

We now define the agents induced by the AC program $ACP_{otc}$ given
above according to Definition~\ref{def:ind-ag}.
\begin{definition}
Given the AC program $\ACP_{otc}
= \tup{\D_{otc},U_{otc},\Sigma_{otc}}$, the agents $A_c$, $A_m$ and
$A_s$ induced by $\ACP_{otc}$ are defined as follows:
\begin{itemize} 
\item $A_c = \tup{\D_c, L_c, Act_c, Pr_c}$, where
\myi $\D_c$ is as above;
\myii $L_c =\D_c(U_{otc})$;
\myiii $Act_c = \Omega_c = \{ {\it createPO(prod\_code,offer)}, {\it submitPO(po\_id)},  {\it pay(po\_id)},  {\it deletePO(po\_id)} \}$; and
\myiv $\alpha(\vec u) \in Pr_c(l_c)$ iff $l_c \models \pi(\vec v)$ for $\alpha(\vec u)=\tup{\pi(\vec
	v),\psi(\vec w)}$.
\item $A_m = \tup{\D_m, L_m, Act_m, Pr_m}$, where
 \myi $\D_m$ is as above;
\myii $L_m = \D_m(U_{otc})$;
\myiii $Act_m = \Omega_m = \{ {\it createMO(po\_id,price)}, {\it doneMO(mo\_id)},  {\it shipPO(po\_id)},  {\it deleteMO(mo\_id)} \}$; and
\myiv $\alpha(\vec u) \in Pr_m(l_m)$ iff $l_m \models \pi(\vec v)$ for $\alpha(\vec u)=\tup{\pi(\vec
	v),\psi(\vec w)}$.
\item $A_s = \tup{\D_s, L_s, Act_s, Pr_s}$, where
 \myi $\D_s$ is as above;
\myii $L_s = \D_s(U_{otc})$;
\myiii $Act_s = \Omega_s = \{ {\it acceptMO(mo\_id)}, {\it rejectMO(mo\_id)},  {\it shipMO(mo\_id)} \}$; and
\myiv $\alpha(\vec u) \in Pr_s(l_s)$ iff $l_m \models \pi(\vec v)$ for $\alpha(\vec u)=\tup{\pi(\vec
	v),\psi(\vec w)}$.	
\end{itemize}
\end{definition}
 
By the definition of $A_m$ we can see that $createMO(po\_id,
 price) \in Pr_m(l_m)$ if and only if the interpretation $l_m(PO)$ of
 the relation $PO$ in the local state $l_m$ contains a tuple $\tup{
 po\_id, pc, o, {\sf prepared} }$ for some product $pc$ and offer $o$;
 while $doneMO(mo\_id) \in Pr_m(l_m)$ iff $l_m(MO)$ contains a tuple
 in the interpretation $l_m(MO)$ with id $mo\_id$ and status ${\sf
 preparation}$. It can also be checked that, in line with our
 discussion in Section~\ref{sec:aqis}, a full version of the function
 $\tau_{otc}$ given above can easily encode the artifacts' lifecycles
 as given in Figure~\ref{fig:lifecycles}.

We can now define the \aqis\ generated by the set of agents $\{ A_c, A_m ,
A_s\}$ according to Definition~\ref{def:ind-aqis}.
\begin{definition}
Given the AC program $\ACP_{otc}$ and the set $Ag = \set{A_c, A_m, A_s}$ of
agents induced by $\ACP_{otc}$, the \aqis\ induced by $\ACP_{otc}$ is
the tuple $\P_{otc} =\tup{\S_{otc}, U_{otc}, s^{0}_{otc},\tau_{otc}}$, where: 
\begin{itemize}
\item $\S_{otc} \subseteq L_c \times L_m \times L_s$ is the set
of \emph{reachable states};
\item $U_{otc}$ is the interpretation domain;
\item $s^0_{otc} = \tup{l_{c0}, l_{m0}, l_{s0}}$ is the initial global
state, where the only non-empty relation are {\it Products} and {\it Materials}; 
\item $\tau_{otc}$ is the global transition function defined according to Def.~\ref{def:ind-aqis}.
\end{itemize}
\end{definition}

As an example we give a snippet of the transition function
$\tau_{otc}$ by considering the global action $\alpha(\vec u)
= \langle createPO(pc), doneMO(m), acceptMO(m') \rangle$ enabled by
the respective protocols in a global state $s$. By the definition of the
actions $createPO(pc)$, $doneMO(m)$, and $acceptMO(m')$ we have that
$l_i(s) \in Pr_i$ for $i \in \{c,m, s\}$ implies that the {\it
Products} relation contains information about the product $pc$.
Also, the interpretation of the relation $MO$ contains the tuples
$\tup{m, p, pr, {\sf preparation} }$ and $\tup{m', p', pr', {\sf
submitted} }$ for some products $p$ and $p'$.

By the definition of $\tau_{otc}$ it follows that for every
$s' \in \S_{otc}$, $s \xrightarrow{\alpha(\vec u)} s'$ implies that
$D_s \oplus
D_{s'} \models \psi_{createPO}(pc) \land \psi_{doneMO}(m) \land \psi_{acceptMO}(m')$,
that is,
\begin{eqnarray*}
 D_s \oplus D_{s'} & \models & 
\exists
                          id, b~(PO'(id,pc, b, {\sf prepared})
                          \land Products(pc, b) \land\\
& &  \forall id', p, o, s~(PO(id', p, o,s) \ra  id \neq id')) \land \\
& & 			\forall w,p,pr,s~\big(
					(w\neq m \ra(MO(w,p,pr,s)\lra
                                MO'(w,p,pr,s))) \land\\ 
& & (MO(m,p,pr,s) \ra 						(MO'(m,p,pr,{\sf submitted})\land\\
& &   (s\neq{\sf submitted}\ra\lnot MO'(m,p,pr,s))))
				\big) \land\\	
& & 	\forall w,p,pr,s~\big(
					(w\neq m' \ra(MO(w,p,pr,s)\lra
                                MO'(w,p,pr,s))) \land\\ 
& & (MO(m',p,pr,s) \ra 						(MO'(m',p,pr,{\sf accepted})\land\\
& &	(s\neq{\sf accepted}\ra\lnot MO'(m',p,pr,s))))
				\big)
\end{eqnarray*}
Hence, the interpretation of the relation {\it PO} in $D_{s'}$ extends
$D_{s}(PO)$ with the tuple $\tup{id, pc, b, {\sf prepared}}$, where
$id$ is a fresh id. The tuples for the material orders $m$ and $m'$
are updated in $D_{s'}(MO)$ by becoming $\tup{m, p, pr, {\sf
submitted} }$ and $\tup{m', p', pr', {\sf accepted} }$,
respectively. In view of the second condition on $\tau_{otc}$ in
Definition~\ref{def:ind-aqis}, no other elements are changed in the
transition. Finally, notice that these extensions are indeed the
interpretations of {\it PO} and {\it MO} in $D_{s'}$. Thus, the
operational semantics satisfies the intended meaning of actions.



We can now investigate properties of the AC program $ACP_{otc}$ by
using specifications in FO-CTLK.
For instance, the following formula specifies that the manufacturer
$m$ knows that each material order MO has to match a corresponding
purchase order PO:
\begin{eqnarray*}
\varphi_{match} & = & AG~\forall id, pc~(\exists pr, s~ MO(id,
pc, pr, s)   \to K_m \exists o,s' PO(id, pc, o,s'))
\end{eqnarray*}

The next specification states that given a material order MO, the
customer will eventually know that the corresponding PO will be shipped.
\begin{eqnarray*}
\varphi_{\it fulfil} & = & AG~\forall id, pc~(\exists pr, s~MO(id, pc, pr, s) \to EF~K_c \exists  o~PO(id, pc, o, {\sf shipped}))
\end{eqnarray*}

Further, we may be interested in checking whether budget and costs are
always kept secret from the supplier $s$ and the customer $c$
respectively, and whether the customer (resp., the supplier) knows
this fact:
\begin{eqnarray*}
\varphi_{budget} & = & K_c~\forall pc~AG~\neg \exists b~K_s~Products(pc,b)\\
\varphi_{cost} & = & K_s~\forall mc~AG~\neg \exists c~K_c~Materials(mc,c)
\end{eqnarray*}




Other interesting specifications describing properties of the artifact
system and the agents operating in it can be similarly formalised in
FO-CTLK, thereby providing the engineer with a valuable tool to assess
the implementation.

We now proceed to exploit the methodology of Section~\ref{sec:results}
to verify the AC program $ACP_{otp}$. We use $\varphi_{match}$ as an
example specification; analogous results can be obtained for other
formulas. Observe that according to Definition~\ref{def:ind-aqis} the
\aqis\ induced by $ACP_{otp}$ has infinitely many states. 

We assume two interpretations for the relations {\it Products} and
{\it Materials}, which determine an initial state $D_0$. Consider the
maximum number $max$ of parameters and the constants $C_{\Omega}$ in
the operations in $\Omega_c$, $\Omega_m$ and $\Omega_s$.  In the case
under analysis we have that $max = 2$. We earlier remarked that
formulas such as $\phi_b$ in the postcondition of actions force
the \aqis\ $\P_{otc}$ corresponding to $ACP_{otc}$ is bounded. Here we
have that $\P_{otc}$ is $b$-bounded. According to
Corollary~\ref{cor:preservation}, we can therefore consider a finite
domain $U'$ such that
\begin{eqnarray*}
 U'  & \supseteq & D_0 \cup C_{\Omega} \cup
\const(\varphi_{match})\\
 & & D_0(Products) \cup D_0(Materials) \cup C_{\Omega}
\end{eqnarray*}
and such that 
\begin{eqnarray*}
|U'| & \geq & 2b + |D_0| + |C_{\Omega}| + |\const(\varphi_{match})| + max\\
 &  =  & 2b + |D_0| + |C_{\Omega}| + 2
\end{eqnarray*}
For instance, we can consider any subset $U'$ of $U_{otc}$ satisfying
the conditions above.  Given that $U'$ satisfies the hypothesis of
Theorem~\ref{th:spec-abstr}, it follows that the AC program
$ACP_{otc}$ over $U_{otc}$ satisfies $\varphi_{match}$ if and only if
$ACP_{otc}$ over $U'$ does. But the \aqis\ induced by the latter is a
finite-state system, which can be constructively built by running the
AC program $ACP_{otc}$ on the elements in $U'$. Thus,
$ACP_{otc} \models \varphi_{match}$ is a decidable instance of model
checking that can be therefore solved by means of standard techniques.

A manual check on the finite model indeed reveals that
$\varphi_{match}, \varphi_{budget}$ and $\varphi_{cost}$ are satisfied in the
finite model, whereas $\varphi_{\it fulfil}$ is not. 
By Corollary~\ref{cor:preservation} the \aqis\ $\P_{otc}$ induced by
$ACP_{otp}$ satisfies the same specifications. Hence, in view of
Definition~\ref{def:acp-satisfaction}, we conclude that the
artifact-centric program $ACP_{otp}$ satisfies
$\varphi_{match}, \varphi_{budget}$ and $\varphi_{cost}$ but does not
satisfy $\varphi_{\it fulfil}$. This is entirely in line with our
intuitions of the scenario.




  \label{sec:toy}




\section{Conclusions and Future Work}\label{sec:concl}

In this paper we put forward a methodology for verifying agent-based
artifact-centric systems. We proposed \aqis, a novel semantics
incorporating first-order features, that can be used to reason about
multi-agent systems in an artifact-centric setting. We observed that
the model checking problem for these structures against specifications
given in a first-order temporal-epistemic logic is undecidable
and proceeded to identify suitable fragments for which decidability
can be retained.

We identified two orthogonal solutions to this issue. In the former we
operated a restriction to the specification language and showed that,
by limiting ourselves to sentence-atomic temporal-epistemic
specifications, infinite-state, bounded \aqis\ admit finite
abstractions. In the latter we kept the full first-order
temporal-epistemic logic but identified the noteworthy subset of
uniform \aqis. In this setting we showed that bounded uniform \aqis\
admit finite abstractions. The abstractions we identified in each
setting depend on novel notions of bisimulation at first-order that
we proposed. 

We explored the complexity of the model checking problem in this
context and showed this to be EXPSPACE-complete. While this is
obviously a hard problem, we need to consider that these
are first-order structures which normally lead to undecidable
problems. We were also reassured by the fact that the abstract
interpretation domain is actually linear in the size of the bound
considered.

Mindful of the practical needs for verification in artifact-centric
systems, we then explored how finite abstractions can actually be
built. To this end, rather than investigating one specific
data-centric language, we defined a general class of declarative
artifact-centric programs. We showed that these systems admit uniform
\aqis\ as their semantics. Under the assumption of bounded systems we
showed that model checking these multi-agent system programs is
decidable and gave a constructive procedure operating on bisimilar,
finite models. While the results are general, they can be instantiated
for various artifact-centric languages. For
instance~\cite{BelardinelliLP12b} explores finite abstractions of GSM
programs by using these results. 

We exemplified the methodology put forward on a use-case consisting of
several agents purchasing and delivering products. While the system
has infinitely many states we showed it admits a 
finite abstraction that can be used to verify a variety of
specifications on the system.

A question left open in the present paper is whether the uniform
condition we provided is tight. While we showed this to be a
sufficient condition, we did not explore whether this is necessary for
finite abstractions or whether more general properties can be
given. In this context it is of interest that artifact-centric
programs generate uniform structures. Also, it will be worthwhile to
explore whether a notion related to uniformity can be applied to other
domains in AI, for example to retain decidability of specific
calculi. This would appear to be the case as preliminary studies in
the Situation Calculus demonstrate~\cite{degiacomo-etal-kr12}.

On the application side, we are also interested in exploring ways to
use the results of this paper to build a model checker for
artifact-centric MAS. Previous efforts in this area,
including~\cite{GonzalezGL12}, are limited to finite state systems. It
would therefore be of great interest to construct finite abstractions
on the fly to check practical e-commerce scenarios such as the one
here discussed.


\bibliographystyle{theapa}
\bibliography{jair}
\end{document}